\documentclass[10pt, conference]{IEEEtran}
\usepackage{algorithm2e}
\usepackage{amsmath}
\usepackage{amssymb}
\usepackage{amsfonts}
\usepackage{amscd}
\usepackage{mathrsfs}
\usepackage[final]{graphicx}
\usepackage{graphicx}
\usepackage{color}
\usepackage{url}

\newtheorem{theorem}{Theorem}
\newtheorem{lemma}{Lemma}
\newtheorem{definition}{Definition}

\newtheorem{example}{Example}
\newtheorem{remark}{Remark}

\begin{document}

\title{Minimizing the Complexity of Fast Sphere Decoding of STBCs}

\author{
\authorblockN{G. R. Jithamithra and B. Sundar Rajan,\\}
\authorblockA{Dept of ECE, Indian Institute of Science, \\
Bangalore 560012, India\\
Email:\{jithamithra,bsrajan\}@ece.iisc.ernet.in\\
}
}
%\author{G. R. Jithamithra and B. Sundar Rajan, Senior Member, IEEE}
% make the title area
\maketitle

%%%%%%%%%%%%%%%%%%%%%%%%%%%%%%%%%%%%%%%%%%%%%
% Abstract
%%%%%%%%%%%%%%%%%%%%%%%%%%%%%%%%%%%%%%%%%%%%%

\begin{abstract}
Decoding of linear space-time block codes (STBCs) with sphere-decoding (SD) is well known. A fast-version of the SD known as fast sphere decoding (FSD) has been recently studied by Biglieri, Hong and Viterbo. Viewing a linear STBC as a vector space spanned by its defining weight matrices over the real number field, we define a quadratic form (QF), called the Hurwitz-Radon QF (HRQF), on this vector space and give a QF interpretation of the FSD complexity of a linear STBC. It is shown that the FSD complexity is only a function of the weight matrices defining the code and their ordering, and not of the channel realization (even though the equivalent channel when SD is used depends on the channel realization) or the number of receive antennas. It is also shown that the FSD complexity is completely captured into a single matrix obtained from the HRQF. Moreover, for a given set of weight matrices, an algorithm to obtain a best ordering of them leading to the least FSD complexity is presented. The well known classes of low FSD complexity codes (multi-group decodable codes, fast decodable codes and fast group decodable codes) are presented in the framework of HRQF.
\end{abstract}

%%%%%%%%%%%%%%%%%%%%%%%%%%%%%%%%%%%%%%%%%%%%%
% Section 1 - Introduction 
%%%%%%%%%%%%%%%%%%%%%%%%%%%%%%%%%%%%%%%%%%%%%

\section{Introduction \& Preliminaries}
\label{sec1}
Consider a minimal-delay space-time coded Rayleigh quasi-static flat fading MIMO channel with full channel state information at the receiver (CSIR). The input output relation for such a system is given by
\begin{equation}
\label{system_model}
\textbf{Y} = \textbf{H}\textbf{X} + \textbf{N},
\end{equation}
where $\textbf{H} \in \mathbb{C}^{n_{r} \times n_{t}}$ is the channel matrix and $\textbf{N} \in \mathbb{C}^{n_{r} \times n_{t}}$ is the additive noise. Both $\textbf{H}$ and $\textbf{N}$ have entries that are i.i.d. complex-Gaussian with zero mean and variance 1 and $N_{0}$ respectively. The transmitted codeword is $\textbf{X} \in \mathbb{C}^{n_{t} \times n_{t}}$ and $\textbf{Y} \in \mathbb{C}^{n_{r} \times n_{t}}$ is the received matrix. The ML decoding metric to minimize over all possible values of the codeword $\textbf{X},$ is
\begin{equation}
\label{ML}
\textbf{M}\left( \textbf{X}\right) = \parallel \textbf{Y} - \textbf{H}\textbf{X}\parallel_{F}^{2}.
\end{equation}

%%%%%%%%%%%%%%%%%%
\begin{definition}
\label{ld_stbc_def}
A linear STBC \cite{HaH}: A linear STBC $\mathcal{C}$ over a real (1-dimensional) signal set $\mathcal{S}$, is a finite set of $n_{t} \times n_t$ matrices, where any codeword matrix belonging to the code $\mathcal{C}$ is obtained  from,
\begin{equation}
\label{ld_stbc}
\textbf{X}\left( x_{1}, x_{2}, . . . , x_{K} \right) ~=~ \sum_{i = 1}^{K} x_{i}\textbf{A}_{i}, 
\end{equation}
by letting the real variables $x_{1}, x_{2} . . . , x_{K}$ take values from a real signal set $\mathcal{S},$ where $\textbf{A}_{i}$ are fixed $n_{t} \times n_t$ complex matrices defining the code, known as the weight matrices. The rate of this code is $\frac{K}{2n_t}$ complex symbols per channel use.
\end{definition}
%%%%%%%%%%%%%%%%%%%%

We are interested in linear STBCs, since they admit Sphere Decoding (SD) \cite{ViB} which is a fast way of decoding for the variables. A further simplified version of the SD known as the fast sphere decoding (FSD) \cite{BHV} (also known as conditional ML decoding) was studied by Biglieri, Hong and Viterbo. The quadratic form (QF) approach  has been used in the context of STBCs in \cite{UnM} to determine whether Quaternion algebras or Biquaternion algebras are division algebras, an aspect dealing with the full diversity of the codes. This approach has not been fully exploited to study the other characteristics of  STBCs. In this paper, we use this approach to study the fast sphere decoding (FSD) complexity of STBCs (a formal definition of this complexity is given in Subsection \ref{FSDC}).
  
Designing STBCs with low decoding complexity has been studied widely in the literature. Orthogonal designs with single symbol decodability were proposed in \cite{TJC}, \cite{Li}, \cite{TiH}. For STBCs with more than two transmit antennas, these came at a cost of reduced transmission rates. To increase the rate at the cost of higher decoding complexity, multi-group decodable STBCs were introduced in \cite{DYT}, \cite{KaR}, \cite{KaR1}. 
Fast decodable codes (codes that admit FSD) have reduced SD complexity owing to the fact that a few of the variables can be decoded as single symbols or in groups if we condition them with respect to the other variables. Fast decodable codes for asymmetric systems using division algebras have been recently reported \cite{VHO}. Golden code and Silver code are also examples of fast decodable codes as shown in \cite{PaR} and \cite{SiB}. The properties of fast decodable codes and multi-group decodable codes were combined and a new class of codes called fast group decodable codes were studied in \cite{RGYS}.

%%%--------------------------------------------------
\subsection{Hurwitz-Radon Quadratic Form}
\label{HRQF}
In this subsection we define the Hurwitz Radon quadratic form (HRQF) on any STBC. We first recall some basics about quadratic forms. More details can be seen in  \cite{La}.

\begin{definition}
\label{quad_form_def} 
Let $F$ be a field with characteristic not 2, and $V$ be a finite dimensional $F$-vector space. A quadratic form on $V$ is defined as a map $Q:V \longrightarrow F$ such that it satisfies the following properties.
\begin{itemize}
\item $Q\left(a \textbf{v}\right) = a^{2}Q\left( \textbf{v}\right)$ for all $ \textbf{v} \in V$ and all $a \in F$. \item The map $B\left( \textbf{v}, \textbf{w}\right) = \frac{1}{2}\left[ Q\left( \textbf{v} + \textbf{w}\right) - Q\left( \textbf{v}\right) - Q\left( \textbf{w}\right)\right] $ for all $ \textbf{v}, \textbf{w} \in V$ is bilinear and symmetric. \end{itemize} 
\end{definition} 

If we consider $V$ as an $n$-dimensional vector space over $F$, then we can also consider the quadratic form as a homogeneous polynomial of degree two, i.e., for $1\leq i,j \leq n$, we have scalars $m_{ij}$ such that
\begin{equation}
\label{qf_matrix_def}
Q\left(\textbf{v}\right) = Q\left(v_{1}, v_{2}, ..., v_{n}\right) = \sum_{i,j=1}^{n} m_{ij}v_{i}v_{j}
\end{equation}
for all $\textbf{v} = \left[v_{1}, ..., v_{n}\right] \in V$. Hence, we can associate a matrix $ \textbf{M} = \left(m_{ij}\right) $ with the quadratic form such that $Q\left(\textbf{v}\right) = \textbf{v} \textbf{M} \textbf{v}^{T}.$
%%%%%%%%%%%%%%%%%%%%
\begin{definition}
\label{hrqf_def}
The Hurwitz Radon quadratic form is a map from the STBC $ \mathcal{C} = \left\lbrace \textbf{X} = \sum_{i = 1}^{K} x_{i}\textbf{A}_{i} \right\rbrace  $ to the field of real numbers $ \mathbb{R}$, i.e., $Q: \mathcal{C} \longrightarrow \mathbb{R}$ given by
\begin{equation}
\label{hrqf_eq}
Q\left( \textbf{X}\right) = \sum_{1\leq i \leq j \leq K} x_{i} x_{j} d_{ij},
\end{equation}
where $ \textbf{X}$ is an element of the STBC and $$d_{ij} = \parallel \textbf{A}_{i} \textbf{A}_{j}^{H} + \textbf{A}_{j} \textbf{A}_{i}^{H}\parallel_{F}^{2}.$$
\end{definition}
%%%%%%%%%%%%%%%%%%%%%%%%%

\begin{theorem}
\label{hrqf_is_qf_thm}
The map defined by \eqref{hrqf_eq} is a quadratic form.
\end{theorem}
\begin{proof}
The map $Q$ needs to satisfy the conditions as defined in Definition \ref{quad_form_def}. We have
\begin{equation*}
\label{hrqf_thm_eq1}
Q\left(a \textbf{X}\right) = \sum_{i,j} ax_{i}.ax_{j}.d_{ij} = a^{2}\sum_{i,j} x_{i}.x_{j}.d_{ij} = a^{2}Q\left(\textbf{X}\right)
\end{equation*}
and
\begin{equation*}
\label{hrqf_thm_eq2}
B\left( \textbf{X}, \textbf{Y}\right) = \frac{1}{2}\left[ Q\left( \textbf{X} + \textbf{Y}\right) - Q\left( \textbf{X}\right) - Q\left( \textbf{Y}\right)\right]
\end{equation*}
should be bilinear and symmetric where $ \textbf{X} = \sum_{i = 1}^{K} x_{i}\textbf{A}_{i} $ and $ \textbf{Y} = \sum_{i = 1}^{K}\left( y_{i}\textbf{A}_{i} \right)$. Substituting and simplifying, we get
\begin{equation*}
\label{hrqf_thm_eq3}
B\left( \textbf{X}, \textbf{Y}\right) = \frac{1}{2} \sum_{i,j} \left[ x_{i}y_{i}d_{ii} + \left(x_{i}y_{j} + x_{j}y_{i} \right)d_{ij} \right].
\end{equation*}
It is clearly seen that this map is bilinear and symmetric.
\end{proof}
%%%%

We can associate a matrix with the HRQF. If we define the matrix $ \textbf{M} = \left(m_{ij}\right) $ where $i, j = 1, 2, ..., K$ such that $m_{ij} = d_{ij}$, then we can write the HRQF as $Q\left(\textbf{x}\right) = \textbf{x} \textbf{M} \textbf{x}^{T},$ where $ \textbf{x} = \left[x_{1} ~ x_{2} ~ ... ~ x_{K}\right] $. Notice that $ \textbf{M}$ is a symmetric matrix and $m_{ij} = 0$ if and only if $\textbf{A}_{i} \textbf{A}_{j}^{H} + \textbf{A}_{j} \textbf{A}_{i}^{H} = \textbf{0}$.

The following example shows that the FSD complexity depends on the ordering of the weight matrices or equivalently the ordering of the variables. 
%%%%%%%%%%%%%%%%%%%%%%%%%%%%%%%%%%%%%%%%%%%%%%%
\begin{example}
\label{silver_code}
Let us consider the Silver code given by:
\begin{equation}
\label{silver_code_eqn}
\textbf{X} = \textbf{X}_{a}\left(s_{1}, s_{2}\right) + \textbf{T}\textbf{X}_{b}\left(z_{1}, z_{2}\right),
\end{equation}
where $\textbf{X}_{a}$ and $\textbf{X}_{b}$ take the Alamouti structure, and
$\textbf{X}_{a}\left(s_{1}, s_{2}\right) = \left[\begin{array}{rr}
s_{1} & -s_{2}^{*}\\
s_{2} & s_{1}^{*}\\
\end{array}\right],$
 $\textbf{X}_{b}\left(z_{1}, z_{2}\right) = \left[\begin{array}{rr}
z_{1} & -z_{2}^{*}\\
z_{2} & z_{1}^{*}\\
\end{array}\right],$ 
$\textbf{T} = \left[\begin{array}{cc}
1 & 0\\
0 & -1\\
\end{array}\right]$ 
and $\left[ z_{1}, z_{2}\right]^{T} = \textbf{U}\left[ s_{3}, s_{4}\right]^{T},$ where $\textbf{U}$ is a unitary matrix chosen to maximize the minimum determinant and is given by 
$\textbf{U} = \frac{1}{\sqrt{7}}\left[\begin{array}{cc}
1+j & -1+2j\\
1+2j & 1-j\\
\end{array}\right].$

Let all the variables take values from a signal set of cardinality $M$. If we order the variables (and hence the weight matrices) as  $\left[s_{1I} , s_{1Q} , s_{2I} , s_{2Q} , s_{3I} , s_{3Q} , s_{4I} , s_{4Q} \right]$, then the $ \textbf{R}$ matrix for SD has the following structure

{\scriptsize
\begin{equation*}
\textbf{R} = \left[\begin{array}{cccccccc}
t & 0 & 0 & 0 & t & t & t & t\\
0 & t & 0 & 0 & t & t & t & t\\
0 & 0 & t & 0 & t & t & t & t\\
0 & 0 & 0 & t & t & t & t & t\\
0 & 0 & 0 & 0 & t & 0 & 0 & 0\\
0 & 0 & 0 & 0 & 0 & t & 0 & 0\\
0 & 0 & 0 & 0 & 0 & 0 & t & 0\\
0 & 0 & 0 & 0 & 0 & 0 & 0 & t\\
\end{array}\right],
\end{equation*}
}

\noindent
where $t$ denotes non zero entries. We can clearly see that the Silver code admits fast decoding with this ordering with FSD complexity $M^{5}$. However, if we change the ordering to $ \left[s_{1I} , s_{1Q} , s_{4I} , s_{2Q} , s_{3I} , s_{3Q} , s_{2I} , s_{4Q} \right]$, then the $ \textbf{R}$ matrix for SD has the following structure 

{\scriptsize
\begin{equation*}
\textbf{R} = \left[\begin{array}{cccccccc}
t & 0 & t & 0 & t & t & 0 & t\\
0 & t & t & 0 & t & t & 0 & t\\
0 & 0 & t & t & t & t & t & 0\\
0 & 0 & 0 & t & t & t & t & t\\
0 & 0 & 0 & 0 & t & t & t & t\\
0 & 0 & 0 & 0 & 0 & t & t & t\\
0 & 0 & 0 & 0 & 0 & 0 & t & t\\
0 & 0 & 0 & 0 & 0 & 0 & 0 & t\\
\end{array}\right],
\end{equation*}
}

\noindent
where $t$ denotes non zero entries. With this ordering, the FSD complexity increases to $M^{7}$. 
\end{example}
%%%%%%%%%%%%%%%%%%%%%%%%%%%%%%%%%%%%%%%%%%%%%%%%%%%%%%%%%%%%%%%%%%%%
%%%%%%%%%%%%%%%%%%%%%%%%%%%%%%%%%%%%%%%%%%%%%%%%%%%%%%%%%%%%%%%%%%%%

The contributions of this paper are as follows: 
\begin{itemize}
\item We give a formal definition of the FSD complexity of a linear STBC (Subsection \ref{FSDC}.)
\item With the help of HRQF, it is shown that the FSD complexity of the code depends only on the weight matrices of the code with their ordering, and not on the channel realization (even though the equivalent channel when SD is used depends on the channel realization) or the number of receive antennas. 
\item A best ordering (not necessarily unique) of the weight matrices provides the least FSD complexity for the STBC. We provide an algorithm to be applied to the HRQF matrix which outputs a best ordering.
\end{itemize}

%%%%%%%%%%%%%%%%%%%%%%%%%%%%
The remaining of the paper is organized as follows: In Section \ref{sec2} the known classes of low ML decodable codes, the system model and the formal definition of the FSD complexity of a linear STBC are given. In Section \ref{sec3}, we  show that the FSD complexity depends completely on the HRQF and not on the channel realization or the number of receive antennas. In Section \ref{sec4}, we present an algorithm to modify the HRQF matrix in order to obtain a best ordering of the weight matrices to obtain the least FSD complexity. Concluding remarks constitute Section \ref{sec5}. 

\indent \textit{Notations:} Throughout the paper, bold lower-case letters are used to denote vectors and bold upper-case letters to denote matrices. For a complex variable $x$, $x_{I}$ and $x_{Q}$ denote the real and imaginary part of $x$, respectively. The sets of all integers, all real and complex numbers are denoted by $\mathbb{Z}, \mathbb{R}$ and $\mathbb{C}$, respectively. The operation of stacking the columns of $\textbf{X}$ one below the other is denoted by $vec\left(\textbf{X}\right)$. The Kronecker product is denoted by $\otimes$, $\textbf{I}_{T}$ and $\textbf{O}_{T}$ denote the $T \times T$ identity matrix and the null matrix, respectively. For a complex variable $x$, the $\check{\left(\centerdot\right)}$ operator acting on $x$ is defined as follows
\begin{equation*}
\check{x} \triangleq \left[\begin{array}{rr}
x_{I} & -x_{Q}\\
x_{Q} & x_{I}\\
\end{array}\right].
\end{equation*}
The $\check{\left(\centerdot\right)}$ operator can similarly be applied to any matrix $\textbf{X} \in \mathbb{C}^{n \times m}$ by replacing each entry $x_{ij}$ by $\check{x}_{ij}$ , $i = 1, 2, ~\cdots , n, j = 1, 2, ~\cdots , m,$ resulting in a matrix denoted by $\check{\textbf{X}} \in \mathbb{R}^{2n \times 2m}$. Given a complex vector $\textbf{x} = \left[x_{1}, x_{2}, \cdots , x_{n}\right]^{T},$ $\tilde{\textbf{x}}$ is defined as
$\tilde{\textbf{x}} \triangleq \left[x_{1I},x_{1Q},\cdots ,x_{nI},x_{nQ}\right]^{T}.$

%%% Section 2 %%%%%%%%%%%%%%%%%%%%%%%%%%%%%%%%%%%%%%%%%%%%%%%%%%%%%%%%%%%%%%%%%%%%%%%%%%%%%%

\section{System Model and Definition of FSD complexity}
\label{sec2}

For any Linear STBC with variables $x_{1}, x_{2} . . . , x_{K}$ given by (\ref{ld_stbc}), the generator matrix $\textbf{G}$ \cite{BHV} is defined by $\widetilde{vec\left(\textbf{X}\right)} = \textbf{G} \tilde{\textbf{x}},$ 
where $\tilde{\textbf{x}} = \left[x_{1}, x_{2} . . . , x_{K}\right]^{T}$. In terms of the weight matrices, the generator matrix can be written as
\begin{equation*}
\textbf{G} = \left[\widetilde{vec\left(\textbf{A}_{1}\right)} ~ \widetilde{vec\left(\textbf{A}_{2}\right)} ~ \cdots ~ \widetilde{vec\left(\textbf{A}_{K}\right)} ~ \right].
\end{equation*}
Hence, for any STBC,  \eqref{system_model} can be written as
\begin{equation*}
\widetilde{vec\left(\textbf{Y}\right)} = \textbf{H}_{eq}\tilde{\textbf{x}} + \widetilde{vec\left(\textbf{N}\right)},
\end{equation*}
where $\textbf{H}_{eq} \in \mathbb{R}^{2n_{r}n_{t} \times K}$ is given by
$\textbf{H}_{eq} = \left(\textbf{I}_{n_{t}} \otimes \check{\textbf{H}}\right) \textbf{G},$ 
 and 
$\tilde{\textbf{x}} = \left[x_{1}, x_{2} . . . , x_{K}\right],$ 
with each $x_{i}$ drawn from a 1-dimensional (PAM) constellation. Using the above equivalent system model, the ML decoding metric \eqref{ML} can be written as
\begin{equation*}
\textbf{M}\left(\tilde{\textbf{x}}\right) = \parallel \widetilde{vec\left(\textbf{Y}\right)} - \textbf{H}_{eq}\tilde{\textbf{x}}\parallel_{F}^{2}.
\end{equation*}
Using $\textbf{Q}\textbf{R}$ decomposition of $\textbf{H}_{eq}$, we get $\textbf{H}_{eq} = \textbf{Q}\textbf{R}$ where $\textbf{Q} \in \mathbb{R}^{2n_{r}n_{t} \times K}$ is an orthonormal matrix and $ \textbf{R} \in \mathbb{R}^{K \times K}$ is an upper triangular matrix. Using this, the ML decoding metric now changes to 
\begin{equation}
\label{eq_ml_decoding_metric}
\textbf{M}\left(\tilde{\textbf{x}}\right) = \parallel \textbf{Q}^{T}\widetilde{vec\left(\textbf{Y}\right)} - \textbf{R}\tilde{\textbf{x}}\parallel_{F}^{2} = \parallel \textbf{y}^{'} - \textbf{R}\tilde{\textbf{x}}\parallel_{F}^{2}.
\end{equation}
If we have $\textbf{H}_{eq} = \left[ \textbf{h}_{1} \textbf{h}_{2} ..., \textbf{h}_{K}\right] ,$ where $ \textbf{h}_{i}, i \in 1, 2, ... , K$ are column vectors, then the $ \textbf{Q}$ and $ \textbf{R}$ matrices have the following form obtained by the Gram-Schmidt orthogonalization: 
\begin{equation}
\label{q_mat}
\textbf{Q} = \left[ \textbf{q}_{1} ~ \textbf{q}_{2} ~ ... ~ \textbf{q}_{K}\right] ,
\end{equation}
where $ \textbf{q}_{i}, i \in 1, 2, ... , K$ are column vectors, and
\begin{equation}
\label{r_mat_def}
\textbf{R} = \left[\begin{array}{ccccc}
\parallel \textbf{r}_{1} \parallel & \left\langle \textbf{q}_{1}, \textbf{h}_{2}\right\rangle & \left\langle \textbf{q}_{1}, \textbf{h}_{3}\right\rangle & \cdots & \left\langle \textbf{q}_{1}, \textbf{h}_{K}\right\rangle\\
0 & \parallel \textbf{r}_{2} \parallel & \left\langle \textbf{q}_{2}, \textbf{h}_{3}\right\rangle & \cdots & \left\langle \textbf{q}_{2}, \textbf{h}_{K}\right\rangle\\
0 & 0 & \parallel \textbf{r}_{3} \parallel & \cdots & \left\langle \textbf{q}_{3}, \textbf{h}_{K}\right\rangle\\
\vdots & \vdots & \vdots & \ddots & \vdots\\
0 & 0 & 0 & \cdots & \parallel \textbf{r}_{K} \parallel\\
\end{array}\right],
\end{equation}
where $\textbf{r}_{1} = \textbf{h}_{1},~~~~ \textbf{q}_{1} = \frac{\textbf{r}_{1}}{\parallel \textbf{r}_{1} \parallel}$ and for $i = 2, ... K,$
\begin{equation*}
\label{r_mat_entries2}
\textbf{r}_{i} = \textbf{h}_{i} - \sum_{j=1}^{i-1} \left\langle \textbf{q}_{j}, \textbf{h}_{i}\right\rangle \textbf{q}_{j} , ~~~~\textbf{q}_{i} = \frac{\textbf{r}_{i}}{\parallel \textbf{r}_{i} \parallel}.
\end{equation*}
%%%%%%%%%%%%%%%%%%%%%%%%%%%%%%%%%%%%%%%%%%%%%%%%%

\subsection{Multi-group decodability, fast decodability and fast group decodability}

In case of a multi-group decodable STBC, the variables can be partitioned into groups such that the ML decoding metric is decoupled into submetrics such that only the members of the same group need to be decoded jointly. It can be formally defined as \cite{KaR}, \cite{KhR}, \cite{RaR}:

\begin{definition}
\label{multi_group_decodability}
An STBC is said to be $g$-group decodable if there exists a partition of $\left\lbrace 1, 2, ... , K\right\rbrace $ into $g$ non-empty subsets $\Gamma_{1}, \Gamma_{2}, ... , \Gamma_{g}$ such that the following condition is satisfied:
\begin{equation*}
\label{multi_group_dec_cond}
\textbf{A}_{l}\textbf{A}_{m}^{H} + \textbf{A}_{m}\textbf{A}_{l}^{H} = \textbf{0},
\end{equation*}
whenever $l \in \Gamma_{i}$ and $m \in \Gamma_{j}$ and $i \neq j$. 
\end{definition}
If we group all the variables of the same group together in \eqref{eq_ml_decoding_metric}, then the $ \textbf{R}$ matrix for the SD \cite{ViB}, \cite{DCB} in case of multi-group decodable codes will be of the following form:
\begin{equation}
\label{multi_group_r_mat}
\textbf{R} = \left[\begin{array}{cccc}
\Delta_{1} & \textbf{0} & \cdots & \textbf{0}\\
\textbf{0} & \Delta_{2} & \cdots & \textbf{0}\\
\vdots & \vdots & \ddots & \vdots\\
\textbf{0} & \textbf{0} & \cdots & \Delta_{g}\\
\end{array}\right],
\end{equation}
where $\Delta_{i}, i = 1, 2, ..., g$ is a square upper triangular matrix. 

Now, consider the standard SD of an STBC. Suppose the $ \textbf{R}$ matrix as defined in \eqref{r_mat_def} turns out to be such that when we fix values for a set of symbols, the rest of the symbols become group decodable, then the code is said to be fast decodable. Formally, it is defined as follows:
\begin{definition}
\label{fast_decodability}
An STBC is said to be fast SD if there exists a partition of $\left\lbrace 1, 2, ... , L\right\rbrace $ where $L \leq K$ into $g$ non-empty subsets $\Gamma_{1}, \Gamma_{2}, ... , \Gamma_{g}$ such that the following condition is satisfied
\begin{equation}
\label{fast_decode_eq}
\left\langle \textbf{q}_{i}, \textbf{h}_{j} \right\rangle = 0 \left( i < j\right) ,
\end{equation}
whenever $i \in \Gamma_{p}$ and $j \in \Gamma_{q}$ and $p \neq q$ where $ \textbf{q}_{i}$ and $ \textbf{h}_{j}$ are obtained from the $\textbf{Q}\textbf{R}$ decomposition of the equivalent channel matrix $\textbf{H}_{eq} = \left[ \textbf{h}_{1} \textbf{h}_{2} ..., \textbf{h}_{K}\right] = \textbf{Q}\textbf{R}$ with $ \textbf{h}_{i}, i \in 1, 2, ... , K$ as column vectors and $\textbf{Q} = \left[ \textbf{q}_{1} ~ \textbf{q}_{2} ~ ... ~ \textbf{q}_{K}\right]$ with $ \textbf{q}_{i}, i \in 1, 2, ... , K$ as column vectors as defined in \eqref{q_mat}.   
\end{definition}

Hence, by conditioning $K - L$ variables, the code becomes $g$-group decodable. As a special case, when no conditioning is needed, i.e., $L = K$, then the code is $g$-group decodable. The $ \textbf{R}$ matrix for fast decodable codes will have the following form:
\begin{equation}
\label{fast_decodable_r_mat}
\textbf{R} = \left[\begin{array}{cc}
\Delta & \textbf{B}_{1}\\
\textbf{0} & \textbf{B}_{2}\\
\end{array}\right],
\end{equation}
where $\Delta$ is an $L \times L$ block diagonal, upper triangular matrix, $ \textbf{B}_{2}$ is a square upper triangular matrix and $ \textbf{B}_{1}$ is a rectangular matrix.

Fast group decodable codes were introduced in \cite{RGYS}. These codes combine the properties of multi-group decodable codes and the fast decodable codes. These codes allow each of the groups in the multi-group decodable codes to be fast decoded. 
The $ \textbf{R}$ matrix for a fast group decodable code will have the following form:
\begin{equation}
\label{fast_group_decodable_r_mat}
\textbf{R} = \left[\begin{array}{cccc}
\textbf{R}_{1} & \textbf{0} & \cdots & \textbf{0}\\
\textbf{0} & \textbf{R}_{2} & \cdots & \textbf{0}\\
\vdots & \vdots & \ddots & \vdots\\
\textbf{0} & \textbf{0} & \cdots & \textbf{R}_{g}\\
\end{array}\right],
\end{equation}
where each $ \textbf{R}_{i}, i = 1, 2, ..., g$ will have the following form:
\begin{equation}
\label{fast_group_decodable_ri_mat}
\textbf{R}_{i} = \left[\begin{array}{cc}
\Delta_{i} & \textbf{B}_{i_{1}}\\
\textbf{0} & \textbf{B}_{i_{2}}\\
\end{array}\right],
\end{equation}
where $\Delta_{i}$ is an $L_{i} \times L_{i}$ block diagonal, upper triangular matrix, $ \textbf{B}_{i_{2}}$ is a square upper triangular matrix and $ \textbf{B}_{i_{1}}$ is a rectangular matrix.
%The formal definition is given as follows:
%\begin{definition}
%\label{fast_group_decodability}
%An STBC with weight matrices $ \textbf{A}_{1}, ... , \textbf{A}_{K}$ is said to be fast group decodable if it satisfies the following conditions:
%\begin{itemize}
%\item There exists a partition of $\left\lbrace 1, 2, ... , K\right\rbrace $ into $g$ non empty subsets $\Gamma_{1}, \Gamma_{2}, ... , \Gamma_{g}$ such that $\textbf{A}_{l}\textbf{A}_{m}^{H} + \textbf{A}_{m}\textbf{A}_{l}^{H} = \textbf{0}$ whenever $l \in \Gamma_{i}$ and $m \in \Gamma_{j}$ and $i \neq j$. 
%\item In any of the partition $\Gamma_{i}$ with cardinality $K_{i}$, if there exists a partition of $\left\lbrace 1, 2, ... , L_{i}\right\rbrace $ where $L_{i} \leq K_{i}$ into $g_{i}$ non empty subsets $\Upsilon_{i_{1}}, \Upsilon_{i_{2}}, ... , \Upsilon_{i_{g_{i}}}$ such that we have $\left\langle \textbf{q}_{i_{l_{1}}}, \textbf{h}_{i_{l_{2}}} \right\rangle = 0 \left( i_{l_{1}} < i_{l_{2}}\right)$ whenever $i_{l_{1}} \in \Upsilon_{i_{p}}$ and $j \in \Upsilon_{i_{q}}$ and $p \neq q$ where $i = 1, 2, ..., g$. 
%\end{itemize}
%\end{definition}
%%%%%%%%%%%%%%%%%%%%%%%%%%%%%%%%%%%%%%%%%%%%%%%%
The structure of the $ \textbf{R}$ matrix for each of the codes defined above depends upon the ordering of the weight matrices. If we change the ordering of the weight matrices, the $ \textbf{R}$ matrix may lose its structure and no longer exhibit the desirable decoding properties. The Silver code of Example \ref{silver_code} illustrates this aspect. In the following subsection, we give a formal definition of the FSD complexity.
%%%%%%%%%%%%%%%%?????? %%%
\subsection{FSD Complexity of an STBC}
\label{FSDC}
In this section we define the FSD complexity of an STBC. First we consider a single group decodable case and define the FSD complexity for a particular ordering. We then extend the definition to a multi-group decodable STBC.
Let $\mathcal{C}$ be an STBC with the weight matrices $ \textbf{A}_{1}, ... , \textbf{A}_{K},$ where all the variables take values from a signal set of cardinality $M$.
Let $ \textbf{R}$ be the matrix obtained by the $ \textbf{Q} \textbf{R}$ decomposition used for SD. Let the ordering of the weight matrices used to obtain the $\textbf{R}$ matrix be $ \textbf{A}_{1}, ... , \textbf{A}_{K}$.
Denote by $l_{1}$ the number of variables that need to be conditioned, when we use FSD on the $\textbf{R}$ matrix.
After conditioning $l_{1}$ variables, let the rest of the variables be $p_{1}$-group decodable.
In case the code is not fast sphere decodable, then we set $l_{1}$ to the total number of variables in the matrix.
Denote by $\textbf{R}_{i_{1}}$ the submatrix containing the variables of the $i_{1}$-th group after conditioning $l_{1}$ variables in the $\textbf{R}$ matrix and by $n_{i_{1}}$ the number of variables in the $i_{1}$-th group, where $1 \leq i_{1} \leq p_{1}$.
We can use FSD on each of the $\textbf{R}_{i_{1}}$ matrices now. Let us denote by $l_{1, i_{1}}$ the number of variables that need to be conditioned, when we use FSD on the $\textbf{R}_{i_{1}}$ matrix. After conditioning $l_{1, i_{1}}$ variables in $\textbf{R}_{i_{1}}$, let the rest of the variables be $p_{1, i_{1}}$-group decodable.
Let $p_{2} = \max_{i_{1}}\left( p_{1, i_{1}} \right)$.
Denote by $\textbf{R}_{i_{1}, i_{2}}$ the submatrix containing the variables of the $i_{2}$-th group after conditioning $l_{1, i_{1}}$ variables in the $\textbf{R}_{i_{1}}$ matrix and by $n_{i_{1}, i_{2}}$ the number of variables in the $i_{2}$-th group, where $1 \leq i_{2} \leq p_{2}$. It may so happen that $p_{1, i_{1}} < p_{2}$ for some $i_{1}$. In such cases we set $n_{i_{1}, i_{2}} = 0$ for $i_{2} > p_{1, i_{1}}$. We continue this process till we cannot apply FSD any more. This process terminates since there are a finite number of variables and the number of variables are decreasing with each iteration of the FSD. Let the process stop after $\lambda$ such iterations. In the $\left( \lambda - 1\right)$ -th iteration, we have the following:
Let us denote by $l_{1, i_{1}, i_{2}, ..., i_{\lambda - 1}}$ the number of variables that need to be conditioned, when we use FSD on the $\textbf{R}_{i_{1}, i_{2}, ..., i_{\lambda - 1}}$ matrix. After conditioning $l_{1, i_{1}, i_{2}, ..., i_{\lambda - 1}}$ variables in $\textbf{R}_{i_{1}, i_{2}, ..., i_{\lambda - 1}}$, let the rest of the variables be $p_{1, i_{1}, i_{2}, ..., i_{\lambda - 1}}$-group decodable.
Let $p_{\lambda} = \max_{i_{1}, i_{2}, ..., i_{\lambda - 1}}\left( p_{1, i_{1}, i_{2}, ..., i_{\lambda - 1}} \right)$.
Denote by $\textbf{R}_{i_{1}, i_{2}, ..., i_{\lambda}}$ the submatrix containing the variables of the $i_{\lambda}$-th group after conditioning $l_{1, i_{1}, i_{2}, ..., i_{\lambda - 1}}$ variables in the $\textbf{R}_{i_{1}, i_{2}, ..., i_{\lambda - 1}}$ matrix and by $n_{i_{1}, i_{2}, ..., i_{\lambda}}$ the number of variables in the $i_{\lambda}$-th group where $1 \leq i_{\lambda} \leq p_{\lambda}$.
On the last iteration, we have $l_{1, i_{1}, i_{2}, ..., i_{\lambda}}$ = $n_{i_{1}, i_{2}, ..., i_{\lambda}}$ for $1 \leq i_{\lambda} \leq p_{\lambda}$.

Let $k_{1, i_{1}, i_{2}, ..., i_{j-1}} = \max_{i_{j}}\left( l_{1, i_{1}, i_{2}, ..., i_{j}} + k_{1, i_{1}, i_{2}, ..., i_{j}}\right)$ for $2 \leq j \leq \lambda$, $1 \leq i_{j} \leq p_{j}$ and $k_{1, i_{1}, i_{2}, ..., i_{\lambda}} = 0$ and $k_{1} = \max_{i_{1}}\left(k_{1, i_{1}}\right) $.
\begin{definition}
\label{fsd_complexity_def1}
We define the FSD complexity of a single group decodable STBC for the given ordering to be $M^{l_{1} + k_{1}}$.
\end{definition}

In case of a multi-group decodable code with $g$ groups, the $\textbf{R}$ matrix will be a block diagonal matrix. We then calculate the FSD complexity of each group independently as described above and choose the maximum among them as the FSD complexity of the STBC.
\begin{definition}
\label{fsd_complexity_def2}
We define the FSD complexity of a multi-group decodable STBC with $g$ groups to be $\max_{i} \left( M^{l_{i} + k_{i}} \right),$ where $1\leq i \leq g$.
\end{definition}

We present a few examples to get a better understanding of FSD complexity.
\begin{example}
\label{fsd_example1}
Let the $\textbf{R}$ matrix be of the form: 
{\scriptsize
\begin{equation*}
\textbf{R} = \left[\begin{array}{cccccccc}
a_{11} & a_{12} & 0 & 0 & a_{15} & a_{16} & a_{17} & a_{18}\\
0 & a_{22} & 0 & 0 & a_{25} & a_{26} & a_{27} & a_{28}\\
0 & 0 & a_{33} & a_{34} & a_{35} & a_{36} & a_{37} & a_{38}\\
0 & 0 & 0 & a_{44} & a_{45} & a_{46} & a_{47} & a_{48}\\
0 & 0 & 0 & 0 & a_{55} & 0 & 0 & 0\\
0 & 0 & 0 & 0 & 0 & a_{66} & 0 & 0\\
0 & 0 & 0 & 0 & 0 & 0 & a_{77} & 0\\
0 & 0 & 0 & 0 & 0 & 0 & 0 & a_{88}\\
\end{array}\right].
\end{equation*}
}
Now if we use FSD on this matrix, we will condition $l_{1} = 4$ variables and obtain a $2$-group decodable code. We have  
\begin{equation*}
\textbf{R}_{1} = \left[\begin{array}{cc}
a_{11} & a_{12}\\
0 & a_{22}\\
\end{array}\right], ~~~~
\textbf{R}_{2} = \left[\begin{array}{cc}
a_{33} & a_{34}\\
0 & a_{44}\\
\end{array}\right].
\end{equation*}
Since the number of variables in each of the above matrices are $2$, we have $n_{1} = n_{2} = 2$.
We cannot condition any more variables in either $\textbf{R}_{1}$ or $\textbf{R}_{2}$. So the process stops here and we set $l_{1, 1} = 2$ and $l_{1, 2} = 2$. 
We have
\begin{equation*}
k_{1,1} = k_{1,2} = 0,\\
\end{equation*}
\begin{equation*}
k_{1} = \max\left(l_{1,1} + k_{1,1} , l_{1,2} + k_{1,2}\right) = 2. \\
\end{equation*}
The FSD complexity of this STBC for the given ordering is $M^{l_{1} + k_{1}} = M^{6}$.
\end{example}

\begin{example}
\label{fsd_example2}
Let the $\textbf{R}$ matrix be of the form: 
{\scriptsize
\begin{equation*}
%\textbf{R} = \left[\begin{array}{cccccccccc}
%a_{1,1} & 0 & a_{1,3} & a_{1,4} & 0 & 0 & 0 & 0 & a_{1,9} & a_{1,10}\\
%0 & a_{2,2} & a_{2,3} & a_{2,4} & 0 & 0 & 0 & 0 & a_{2,9} & a_{2,10}\\
%0 & 0 & a_{3,3} & a_{3,4} & 0 & 0 & 0 & 0 & a_{3,9} & a_{3,10}\\
%0 & 0 & 0 & a_{4,4} & 0 & 0 & 0 & 0 & a_{4,9} & a_{4,10}\\
%0 & 0 & 0 & 0 & a_{5,5} & 0 & 0 & a_{5,8} & a_{5,9} & a_{5,10}\\
%0 & 0 & 0 & 0 & 0 & a_{6,6} & 0 & a_{6,8} & a_{6,9} & a_{6,10}\\
%0 & 0 & 0 & 0 & 0 & 0 & a_{7,7} & a_{7,8} & a_{7,9} & a_{7,10}\\
%0 & 0 & 0 & 0 & 0 & 0 & 0 & a_{8,8} & a_{8,9} & a_{8,10}\\
%0 & 0 & 0 & 0 & 0 & 0 & 0 & 0 & a_{9,9} & a_{9,10}\\
%0 & 0 & 0 & 0 & 0 & 0 & 0 & 0 & 0 & a_{10,10}\\
%\end{array}\right]
\textbf{R} = \left[\begin{array}{cccccccccc}
a_{1,1} & 0 & a_{1,3} & a_{1,4} & 0 & 0 & 0 & 0 & t & t\\
0 & a_{2,2} & a_{2,3} & a_{2,4} & 0 & 0 & 0 & 0 & t & t\\
0 & 0 & a_{3,3} & a_{3,4} & 0 & 0 & 0 & 0 & t & t\\
0 & 0 & 0 & a_{4,4} & 0 & 0 & 0 & 0 & t & t\\
0 & 0 & 0 & 0 & a_{5,5} & 0 & 0 & a_{5,8} & t & t\\
0 & 0 & 0 & 0 & 0 & a_{6,6} & 0 & a_{6,8} & t & t\\
0 & 0 & 0 & 0 & 0 & 0 & a_{7,7} & a_{7,8} & t & t\\
0 & 0 & 0 & 0 & 0 & 0 & 0 & a_{8,8} & t & t\\
0 & 0 & 0 & 0 & 0 & 0 & 0 & 0 & t & t\\
0 & 0 & 0 & 0 & 0 & 0 & 0 & 0 & 0 & t\\
\end{array}\right].
\end{equation*}
}
Now if we use FSD on this matrix, we will condition $l_{1} = 2$ variables and obtain a $2$-group decodable code. We have 
\begin{equation*}
\textbf{R}_{1} = \left[\begin{array}{cccc}
a_{1,1} & 0 & a_{1,3} & a_{1,4}\\
0 & a_{2,2} & a_{2,3} & a_{2,4}\\
0 & 0 & a_{3,3} & a_{3,4}\\
0 & 0 & 0 & a_{4,4}\\
\end{array}\right], 
\end{equation*}
\begin{equation*}
\textbf{R}_{2} = \left[\begin{array}{cccc}
a_{5,5} & 0 & 0 & a_{5,8}\\
0 & a_{6,6} & 0 & a_{6,8}\\
0 & 0 & a_{7,7} & a_{7,8}\\
0 & 0 & 0 & a_{8,8}\\
\end{array}\right].
\end{equation*}
Since the number of variables in each of the above matrices are $4$, we have $n_{1} = n_{2} = 4$.
Now if we use FSD on $\textbf{R}_{1}$, we can condition $l_{1, 1} = 2$ variables and obtain a $2$-group decodable code. If we use FSD on $\textbf{R}_{2}$, we can condition $l_{1, 2} = 1$ variable and obtain a $3$-group decodable code. We now have
\begin{equation*}
\textbf{R}_{1,1} = \left[\begin{array}{c}
a_{1,1}\\
\end{array}\right], ~~~~
\textbf{R}_{1,2} = \left[\begin{array}{c}
a_{2,2}\\
\end{array}\right], ~~~~
\end{equation*}
\begin{equation*}
\textbf{R}_{2,1} = \left[\begin{array}{c}
a_{5,5}\\
\end{array}\right], ~~~~
\textbf{R}_{2,2} = \left[\begin{array}{c}
a_{6,6}\\
\end{array}\right], ~~~~
\textbf{R}_{2,3} = \left[\begin{array}{c}
a_{7,7}\\
\end{array}\right],
\end{equation*}
\begin{equation*}
n_{1,1} = n_{1,2} = n_{2,1} = n_{2,2} = n_{2,3} = 1 ~~~ and ~~~ n_{1,3} =  0.
\end{equation*}
We cannot condition any more variables in any of these matrices. So the process stops here and we set 
\begin{equation*}
l_{1,1,1} = l_{1,1,2} = l_{1,2,1} = l_{1,2,2} = l_{1,2,3} = 1.
\end{equation*}
We have 
\begin{equation*}
k_{1,1,1} = k_{1,1,2} = k_{1,2,1} = k_{1,2,2} = k_{1,2,3} = 0,
\end{equation*}
\begin{equation*}
k_{1,1} = \max\left(l_{1,1,1} + k_{1,1,1} , l_{1,1,2} + k_{1,1,2}\right) = 1, \\
\end{equation*}
\begin{equation*}
k_{1,2} = \max\left(l_{1,2,1} + k_{1,2,1} , l_{1,2,2} + k_{1,2,2}, l_{1,2,3} + k_{1,2,3}\right) = 1, \\
\end{equation*}
\begin{equation*}
k_{1} = \max\left(l_{1,1} + k_{1,1} , l_{1,2} + k_{1,2}\right) = 3 .\\
\end{equation*}

The FSD complexity of this STBC for the given ordering is $M^{l_{1} + k_{1}} = M^{5}$.
\end{example}
%%%%%%%%%%%%%
%\begin{definition}
%\label{FSD_def}
%We define the FSD complexity of the STBC $\mathcal{C}$ to be the minimum FSD complexity over all possible orderings of the weight matrices of $\mathcal{C}$.
%\end{definition}
%%%%%%%%%%%

\section{HRQF and FSD complexity}
\label{sec3}

In this section we show that the HRQF matrix is enough to determine the FSD  complexity of an STBC and hence the FSD complexity is independent of the channel matrix realization or the number of receive antennas. Towards this end, we prove that the zeros in the $\textbf{R}$ matrix which determine the FSD complexity are also zeros in the HRQF matrix. First we define an ordered partition of a set. 
%%%%%%%%%%
%\begin{definition}
%\label{ordered_partition_def}
%We call a partition of $\left\lbrace 1, 2, ... , K\right\rbrace $ into $g$ non empty subsets $\Gamma_{1}, \Gamma_{2}, ... , \Gamma_{g}$ an ordered partition if $\left\lbrace 1, ... , \vert \Gamma_{1} \vert\right\rbrace \in \Gamma_{1}$, $\left\lbrace \vert \Gamma_{1} \vert + 1, ... , \vert \Gamma_{1} \vert + \vert \Gamma_{2} \vert\right\rbrace \in \Gamma_{2}$ and so on till $\left\lbrace \sum_{i=1}^{g-1}\vert \Gamma_{i} \vert + 1, ... , \sum_{i=1}^{g}\vert \Gamma_{i} \vert \right\rbrace \in \Gamma_{g}.$
%\end{definition}
%%%%%%%%%%%%%
\begin{definition}
\label{ordered_partition_def}
We call a partition of $\left\lbrace a_{1}, a_{2}, ... , a_{K}\right\rbrace $ into $g$ non-empty subsets $\Gamma_{1}, \Gamma_{2}, ... , \Gamma_{g}$ with cardinalities $K_{1}, K_{2}, ... , K_{g}$ an ordered partition if $\left\lbrace a_{1}, ... , a_{K_{1}}\right\rbrace \in \Gamma_{1}$, $\left\lbrace a_{K_{1} + 1}, ... , a_{K_{1} + K_{2}} \right\rbrace \in \Gamma_{2}$ and so on, till $\left\lbrace a_{\sum_{i=1}^{g-1}K_{i} + 1}, ... , a_{\sum_{i=1}^{g}K_{i}} \right\rbrace \in \Gamma_{g}.$
\end{definition}

%%%%%%%%%%%%
Now we address the class of multi-group decodable codes.

%\subsubsection{HRQF for Multi Group Decodable Codes}
%\label{subsubsec3_3_1}

\begin{lemma}
\label{hrqf_multi_group_lemma}
Consider an STBC $ \mathcal{C} = \sum_{i = 1}^{K} x_{i}\textbf{A}_{i} $. Let $\textbf{M}$ denote the HRQF matrix of this  STBC. If there exists an ordered partition of $\left\lbrace 1, 2, ... , K\right\rbrace $ into $g$ non-empty subsets $\Gamma_{1}, \Gamma_{2}, ... , \Gamma_{g}$ such that $m_{ij} = 0$ whenever $i \in \Gamma_{p}$ and $j \in \Gamma_{q}$ and $p \neq q$, then the code is $g$-group sphere decodable. In other words, the FSD complexity of the STBC is determined by the HRQF matrix. 
\end{lemma}
%%%%
\begin{proof}
Let $\textbf{R}$ be the matrix obtained from the QR decomposition of $\textbf{H}_{eq}$. For the code to be $g$-group sphere decodable, we need to prove that $r_{ij} = 0$, whenever $i \in \Gamma_{p}$ and $j \in \Gamma_{q}$ and $p \neq q$.
We know from \cite{PaR} that if $\textbf{A}_{i} \textbf{A}_{j}^{H} + \textbf{A}_{j} \textbf{A}_{i}^{H} = \textbf{0}$ is satisfied for some $i,j$ then the corresponding columns in the $ \textbf{H}_{eq}$ matrix are orthogonal, i.e., $\langle \textbf{h}_{i}, \textbf{h}_{j} \rangle = 0$. We also know that $m_{ij} = 0$ if and only if $\textbf{A}_{i} \textbf{A}_{j}^{H} + \textbf{A}_{j} \textbf{A}_{i}^{H} = \textbf{0}$. 
Let $L_{p} = \sum_{q=1}^{p} \vert \Gamma_{q} \vert$ where $p = 1, 2, ..., g$ and $L_{0} = 0$.

For any group $ \Gamma_{p}$, we need to prove that $r_{ij} = 0$ for $L_{p-1} + 1 \leq i \leq L_{p}$ and $L_{p} + 1 \leq j \leq K$.
Consider the first group $\Gamma_{1}$. We have $m_{ij} = 0$  for $1 \leq i \leq L_{1}$ and $L_{1} + 1 \leq j \leq K$. We need to prove that the $ \textbf{R}$ matrix has zero entries at the same locations. The proof for this is by induction.\\
For $i=1$ and for any $j \geq L_{1}+1$, 
\begin{equation*}
\langle \textbf{q}_{1}, \textbf{h}_{j} \rangle = \frac{1}{\parallel \textbf{h}_{1}\parallel}\langle \textbf{h}_{1}, \textbf{h}_{j} \rangle = 0
\end{equation*}
since $\textbf{q}_{1} = \frac{1}{\parallel \textbf{h}_{1}\parallel}\textbf{h}_{1}$. Now, let $\langle \textbf{q}_{l}, \textbf{h}_{j} \rangle = 0$ for all $l < i$ for any $i$ such that $1 \leq i \leq L_{1}$. We have, 
\begin{align*}
\langle \textbf{q}_{i}, \textbf{h}_{j} \rangle &= \frac{1}{\parallel \textbf{r}_{i}\parallel}\left[ \langle \textbf{h}_{i} - \sum_{l=1}^{i-1} \langle \textbf{q}_{l}, \textbf{h}_{i} \rangle \textbf{q}_{l}, \textbf{h}_{j} \rangle \right]\\
%\end{equation*}
%\begin{equation*}
&= \frac{1}{\parallel \textbf{r}_{i}\parallel}\left[ \langle \textbf{h}_{i}, \textbf{h}_{j} \rangle - \sum_{l=1}^{i-1} \langle \textbf{q}_{l}, \textbf{h}_{i}\rangle \langle \textbf{q}_{l}, \textbf{h}_{j} \rangle \right] = 0,
\end{align*}
since $\langle \textbf{h}_{i}, \textbf{h}_{j} \rangle = 0$ as $m_{ij}=0$ and $\langle \textbf{q}_{l}, \textbf{h}_{j} \rangle=0$ for $l < i$, by induction hypothesis.  

Now consider the $p$-th group $\Gamma_{p}$. Let the induction hypothesis be true for all groups $1, 2,... p-1$. Consider $ r_{ij}$ where $L_{p-1} + 1 \leq i \leq L_{p}$ and $L_{p} + 1 \leq j \leq K$. We have,
\begin{align*}
r_{ij} &= \langle \textbf{q}_{i}, \textbf{h}_{j} \rangle = \frac{1}{\parallel \textbf{r}_{i}\parallel}\left[ \langle \textbf{h}_{i} - \sum_{l=1}^{i-1} \langle \textbf{q}_{l}, \textbf{h}_{i} \rangle \textbf{q}_{l}, \textbf{h}_{j} \rangle \right]  \\
%\end{equation*}
%\begin{equation*}
&= \frac{1}{\parallel \textbf{r}_{i}\parallel}\left[ \langle \textbf{h}_{i}, \textbf{h}_{j} \rangle - \sum_{l=1}^{i-1} \langle \textbf{q}_{l}, \textbf{h}_{i}\rangle \langle \textbf{q}_{l}, \textbf{h}_{j} \rangle \right] = 0 , 
\end{align*}
since $\langle \textbf{h}_{i}, \textbf{h}_{j} \rangle = 0$ as $m_{ij}=0$ and $\langle \textbf{q}_{l}, \textbf{h}_{j} \rangle=0$ for $l < i$ by the induction hypothesis.
\end{proof}

We now consider an example to illustrate the above lemma. 

\begin{example}
\label{hrqf_multi_group_ex}
Consider the $2 \times 2$ ABBA code given by \cite{TBH}:
\begin{equation*}
\textbf{X} = \left[\begin{array}{cc}
x_{1} + jx_{4} & -x_{2} + jx_{3}\\
-x_{2} + jx_{3} & x_{1} + jx_{4} \\
\end{array}\right],
\end{equation*}
where $x_{i} \in \mathbb{R}$ for $i = 1, 2, 3, 4$. This is a two group decodable code with $\left\lbrace x_{1}, x_{2}\right\rbrace $ belonging to one group and $\left\lbrace x_{3}, x_{4}\right\rbrace $ belonging to the other. The structure of the HRQF matrix $ \textbf{M}$ and the $ \textbf{R}$ matrix are given below with $\left[ x_{1}, x_{2}, x_{3}, x_{4}\right] $ as the ordering of the variables and the weight matrices,
\begin{equation*}
\textbf{M} = \left[\begin{array}{cccc}
t & t & 0 & 0\\
t & t & 0 & 0\\
0 & 0 & t & t\\
0 & 0 & t & t\\
\end{array}\right], ~~~~
\textbf{R} = \left[\begin{array}{cccc}
t & t & 0 & 0\\
0 & t & 0 & 0\\
0 & 0 & t & t\\
0 & 0 & 0 & t\\
\end{array}\right],
\end{equation*}
where $t$ denotes the non-zero entries. As it can be seen, the upper triangular portion of $ \textbf{M}$ matrix and the $ \textbf{R}$ matrix have the same structure. 
\end{example}

Now we move on the class of fast decodable codes. 
\begin{lemma}
\label{hrqf_fast_decode_lemma}
Consider an STBC $ \mathcal{C} = \sum_{i = 1}^{K} x_{i}\textbf{A}_{i} $. Let $\textbf{M}$ denote the HRQF matrix of this STBC. If there exists an ordered partition of $\left\lbrace 1, 2, ... , L\right\rbrace $ where $L \leq K$ into $g$ non-empty subsets $\Gamma_{1}, \Gamma_{2}, ... , \Gamma_{g}$ such that $m_{ij} = 0$ whenever $i \in \Gamma_{p}$ and $j \in \Gamma_{q}$ and $p \neq q$, then the code is fast decodable or conditionally $g$-group decodable. 
\end{lemma}
\begin{proof}
The proof follows from the proof of Lemma \ref{hrqf_multi_group_lemma} by replacing $K$ with $L$.
\end{proof}
%%%%

We now consider an example to illustrate the above lemma. 
\begin{example}
\label{hrqf_fast_decode_ex}
Consider the Silver code as mentioned in Example \ref{silver_code}. If we order the variables (and hence the weight matrices) in the following fashion $ \left[s_{1I} , s_{1Q} , s_{2I} , s_{2Q} , s_{3I} , s_{3Q} , s_{4I} , s_{4Q} \right]$, then the HRQF matrix $ \textbf{M}$ and the $ \textbf{R}$ matrix will have the following structure: 

\begin{equation*}
\textbf{M} = \left[\begin{array}{cccccccc}
t & 0 & 0 & 0 & t & t & t & t\\
0 & t & 0 & 0 & t & t & t & t\\
0 & 0 & t & 0 & t & t & t & t\\
0 & 0 & 0 & t & t & t & t & t\\
t & t & t & t & t & 0 & 0 & 0\\
t & t & t & t & 0 & t & 0 & 0\\
t & t & t & t & 0 & 0 & t & 0\\
t & t & t & t & 0 & 0 & 0 & t\\
\end{array}\right], 
\end{equation*}
\begin{equation*}
\textbf{R} = \left[\begin{array}{cccccccc}
t & 0 & 0 & 0 & t & t & t & t\\
0 & t & 0 & 0 & t & t & t & t\\
0 & 0 & t & 0 & t & t & t & t\\
0 & 0 & 0 & t & t & t & t & t\\
0 & 0 & 0 & 0 & t & 0 & 0 & 0\\
0 & 0 & 0 & 0 & 0 & t & 0 & 0\\
0 & 0 & 0 & 0 & 0 & 0 & t & 0\\
0 & 0 & 0 & 0 & 0 & 0 & 0 & t\\
\end{array}\right],
\end{equation*}
where $t$ denotes the non-zero entries. As it can be seen, the upper triangular portion of the matrix $ \textbf{M}$, has a structure that admits fast decodability which is conditionally $4$-group decodable if considered as the $ \textbf{R}$ matrix.
\end{example}

We now turn to the class of fast group decodable codes. 
%%%%%%%%%%
%\begin{lemma}
%\label{hrqf_fast_group_decode_lemma}
%Consider an STBC $ \mathcal{C} = \sum_{i = 1}^{K} x_{i}\textbf{A}_{i} $. Let $\textbf{M}$ denote the HRQF matrix of this STBC. If there exists an ordered partition of $\left\lbrace 1, 2, ... , K\right\rbrace $ where $L \leq K$ into $g$ non empty subsets $\Gamma_{1}, \Gamma_{2}, ... , \Gamma_{g}$ such that $m_{ij} = 0$ whenever $i \in \Gamma_{p}$ and $j \in \Gamma_{q}$ and $p \neq q$, and each group $\Gamma_{i}$ admits fast decodability i.e., there exists a partition of $\left\lbrace 1, 2, ... , L_{i}\right\rbrace $ where $L_{i} \leq K_{i}$, $K_{i} = \vert \Gamma_{i} \vert$ into $g_{i}$ non empty subsets $\Upsilon_{i_{1}}, \Upsilon_{i_{2}}, ... , \Upsilon_{i_{g_{i}}}$ such that $m_{i_{l_{1}}i_{l_{2}}} = 0$ whenever $i_{l_{1}} \in \Upsilon_{i_{p}}$ and $j \in \Upsilon_{i_{q}}$ and $p \neq q$, $i = 1, 2, ..., g$, then the code is said to be fast group decodable. 
%\end{lemma}
%%%%%%%%%%%%%%%%%%%
\begin{lemma}
\label{hrqf_fast_group_decode_lemma}
Consider an STBC $ \mathcal{C} = \sum_{i = 1}^{K} x_{i}\textbf{A}_{i} $. Let $\textbf{M}$ denote the HRQF matrix of this STBC. If there exists an ordered partition of $\left\lbrace 1, 2, ... , K\right\rbrace $ into $g$ non-empty subsets $\Gamma_{1}, \Gamma_{2}, ... , \Gamma_{g}$ with cardinalities $K_{1}, K_{2}, ... , K_{g}$ such that $m_{ij} = 0$ whenever $i \in \Gamma_{p}$ and $j \in \Gamma_{q}$ and $p \neq q$, and if any group $\Gamma_{i}$ admits fast decodability, i.e., there exists an ordered partition of $\left\lbrace \sum_{l=1}^{i-1}K_{l} + 1, \sum_{l=1}^{i-1}K_{l} + 2, ... , \sum_{l=1}^{i-1}K_{l} + L_{i}\right\rbrace $ where $L_{i} \leq K_{i}$, into $g_{i}$ non-empty subsets $\Upsilon_{i_{1}}, \Upsilon_{i_{2}}, ... , \Upsilon_{i_{g_{i}}}$ such that $m_{rs} = 0$ whenever $r \in \Upsilon_{i_{p}}$ and $s \in \Upsilon_{i_{q}}$ and $p \neq q$, $i = 1, 2, ..., g$, then the code is fast group decodable. 
\end{lemma}
%%%%%%%%%%%%%%%%%%%%
\begin{proof}
The proof follows from the proofs of lemmas \ref{hrqf_multi_group_lemma} and \ref{hrqf_fast_decode_lemma}.
\end{proof}
We now consider an example to illustrate the above lemma. 
\begin{example}
\label{hrqf_fast_group_decode_ex}
Consider the fast group decodable STBC \cite{RGYS} given in \eqref{fgd_stbc} . 
\begin{figure*}
\scriptsize 
\begin{equation}
\label{fgd_stbc}
\textbf{X} = \left[\begin{array}{cccc}
s_{1} + js_{2} + js_{15} + js_{16} + js_{17} & s_{7} + js_{8} + s_{13} + js_{14} & s_{3} + js_{4} + s_{11} + js_{12} & -s_{5} - js_{6} + s_{9} + js_{10}\\
-s_{7} + js_{8} - s_{13} + js_{14} & s_{1} + js_{2} + js_{15} - js_{16} - js_{17} & s_{5} - js_{6} + s_{9} - js_{10} & s_{3} - js_{4} - s_{11} + js_{12}\\
-s_{3} + js_{4} - s_{11} + js_{12} & -s_{5} - js_{6} - s_{9} - js_{10} & s_{1} - js_{2} + js_{15} - js_{16} + js_{17} & s_{7} - js_{8} - s_{13} + js_{14}\\
s_{5} - js_{6} - s_{9} + js_{10} & -s_{3} - js_{4} + s_{11} + js_{12} & -s_{7} - js_{8} + s_{13} + js_{14} & s_{1} - js_{2} + js_{15} + js_{16} - js_{17}\\
\end{array}\right]
\end{equation}
\hrule
\end{figure*} 

Let the ordering of the variables (and hence the weight matrices) be $\left[ s_{1}, s_{2}, ... , s_{17}\right] $. This STBC is two group decodable with $s_{1}$ in one group and $\left\lbrace s_{2}, s_{3}, ..., s_{17}\right\rbrace $ in the other. The second group is conditionally five group decodable. The HRQF matrix $ \textbf{M}$ and the $ \textbf{R}$ matrix are given in \eqref{fgd_stbc_hrqf} and \eqref{fgd_stbc_r} respectively, where $t$ denotes the non zero entries. 

%\begin{figure*}
\scriptsize \addtolength{\arraycolsep}{-1pt}
\begin{equation}
\label{fgd_stbc_hrqf}
\textbf{M} = \left[\begin{array}{ccccccccccccccccc}
t & 0 & 0 & 0 & 0 & 0 & 0 & 0 & 0 & 0 & 0 & 0 & 0 & 0 & 0 & 0 & 0\\
0 & t & 0 & 0 & 0 & 0 & t & t & 0 & 0 & 0 & 0 & t & t & t & t & t\\
0 & 0 & t & 0 & 0 & 0 & t & 0 & t & t & t & 0 & 0 & t & 0 & t & t\\
0 & 0 & 0 & t & 0 & 0 & 0 & t & t & t & 0 & t & t & 0 & 0 & t & t\\
0 & 0 & 0 & 0 & t & 0 & t & t & t & 0 & t & t & 0 & 0 & t & t & 0\\
0 & 0 & 0 & 0 & 0 & t & 0 & 0 & 0 & t & t & t & t & t & t & t & 0\\
0 & t & t & 0 & t & 0 & t & 0 & 0 & t & 0 & t & t & 0 & 0 & t & 0\\
0 & t & 0 & t & t & 0 & 0 & t & 0 & t & t & 0 & 0 & t & 0 & t & 0\\
0 & 0 & t & t & t & 0 & 0 & 0 & t & 0 & 0 & 0 & t & t & t & t & 0\\
0 & 0 & t & t & 0 & t & t & t & 0 & t & 0 & 0 & 0 & 0 & t & t & 0\\
0 & 0 & t & 0 & t & t & 0 & t & 0 & 0 & t & 0 & t & 0 & 0 & t & t\\
0 & 0 & 0 & t & t & t & t & 0 & 0 & 0 & 0 & t & 0 & t & 0 & t & t\\
0 & t & 0 & t & 0 & t & t & 0 & t & 0 & t & 0 & t & 0 & 0 & t & 0\\
0 & t & t & 0 & 0 & t & 0 & t & t & 0 & 0 & t & 0 & t & 0 & t & 0\\
0 & t & 0 & 0 & t & t & 0 & 0 & t & t & 0 & 0 & 0 & 0 & t & t & t\\
0 & t & t & t & t & t & t & t & t & t & t & t & t & t & t & t & t\\
0 & t & t & t & 0 & 0 & 0 & 0 & 0 & 0 & t & t & 0 & 0 & t & t & t\\
\end{array}\right]
\end{equation}
%\hrule
%\end{figure*} 

%\begin{figure*}
\scriptsize
\begin{equation}
\label{fgd_stbc_r}
\textbf{R} = \left[\begin{array}{ccccccccccccccccc}
t & 0 & 0 & 0 & 0 & 0 & 0 & 0 & 0 & 0 & 0 & 0 & 0 & 0 & 0 & 0 & 0\\
0 & t & 0 & 0 & 0 & 0 & t & t & 0 & 0 & 0 & 0 & t & t & t & t & t\\
0 & 0 & t & 0 & 0 & 0 & t & 0 & t & t & t & 0 & 0 & t & 0 & t & t\\
0 & 0 & 0 & t & 0 & 0 & 0 & t & t & t & 0 & t & t & 0 & 0 & t & t\\
0 & 0 & 0 & 0 & t & 0 & t & t & t & 0 & t & t & 0 & 0 & t & t & 0\\
0 & 0 & 0 & 0 & 0 & t & 0 & 0 & 0 & t & t & t & t & t & t & t & 0\\
0 & 0 & 0 & 0 & 0 & 0 & t & t & t & t & t & t & t & t & t & t & t\\
0 & 0 & 0 & 0 & 0 & 0 & 0 & t & t & t & t & t & t & t & t & t & t\\
0 & 0 & 0 & 0 & 0 & 0 & 0 & 0 & t & t & t & t & t & t & t & t & 0\\
0 & 0 & 0 & 0 & 0 & 0 & 0 & 0 & 0 & t & t & t & t & t & t & t & t\\
0 & 0 & 0 & 0 & 0 & 0 & 0 & 0 & 0 & 0 & t & t & t & t & t & t & t\\
0 & 0 & 0 & 0 & 0 & 0 & 0 & 0 & 0 & 0 & 0 & t & t & t & t & t & t\\
0 & 0 & 0 & 0 & 0 & 0 & 0 & 0 & 0 & 0 & 0 & 0 & t & t & t & t & t\\
0 & 0 & 0 & 0 & 0 & 0 & 0 & 0 & 0 & 0 & 0 & 0 & 0 & t & t & t & t\\
0 & 0 & 0 & 0 & 0 & 0 & 0 & 0 & 0 & 0 & 0 & 0 & 0 & 0 & t & t & t\\
0 & 0 & 0 & 0 & 0 & 0 & 0 & 0 & 0 & 0 & 0 & 0 & 0 & 0 & 0 & t & t\\
0 & 0 & 0 & 0 & 0 & 0 & 0 & 0 & 0 & 0 & 0 & 0 & 0 & 0 & 0 & 0 & t\\
\end{array}\right]
\end{equation}
%\hrule
%\end{figure*} 

\end{example}
%%%%

As we have seen from Lemmas \ref{hrqf_multi_group_lemma}, \ref{hrqf_fast_decode_lemma} and \ref{hrqf_fast_group_decode_lemma}, the FSD complexity of the STBC depends only upon the HRQF matrix $ \textbf{M}$ and not on the $ \textbf{H}_{eq}$ matrix, i.e., the FSD  complexity is independent of the channel matrix and the number of receive antennas. It can be completely captured into a single matrix obtained from the set of weight matrices and their ordering.

%%%%%%%%%%%%%%%%%%%%%%%%%%%%%%%%%%%%%%%%%%%%%%%%%
% Algorithm For Ordering the Dispersion Matrices
%%%%%%%%%%%%%%%%%%%%%%%%%%%%%%%%%%%%%%%%%%%%%%%%%

\section{Algorithm for a Best Ordering of the Weight Matrices}
\label{sec4}

As seen in Example \ref{silver_code}, the ordering of weight matrices determines the  FSD complexity of an STBC. We have also seen that the HRQF matrix completely determines the FSD  complexity of an STBC. In this section we present an algorithm that uses the HRQF matrix as an input and manipulates it in order to obtain a best possible ordering of weight matrices. We do so by using row and column permutations of the HRQF matrix. The rows and columns of the HRQF matrix are in one to one correspondence with the ordering of the weight matrices. Hence, if we change the ordering of the weight matrices, the HRQF matrix changes accordingly and vice verse. For example, any transposition in the ordering of the weight matrices will result in swapping the corresponding rows and columns (since HRQF matrix is symmetric) of the HRQF matrix. 
\begin{remark}
\label{hrqf_remark}
Note that we cannot perform such a manipulation on the $\textbf{R}$ matrix since it depends not only on the order of weight matrices but on the channel matrix as well. Also, all the entries of the $\textbf{R}$ matrix do not depict the HR orthogonality of the weight matrices, i.e., the $\left(i, j\right)$-th entry of the matrix may not be zero even if the $i$-th and $j$-th weight matrices are HR orthogonal. Hence, the $\textbf{R}$ matrix needs to be calculated each time the ordering of the weight matrices is changed which is not so in the case of the HRQF matrix.
\end{remark}

The algorithm to get a best possible ordering is given in Algorithm \ref{alg1:alg1} in the next page. 

%%%%%%%%%%%%%%%%%%%%%%%%%%%
{\small
%\begin{figure*}
\begin{table}
\begin{algorithm}[H]
\label{alg1:alg1}
\SetLine
%\linesnumbered
\KwIn{The HRQF matrix - $\textbf{M} = \left( m_{ij}\right)$, the size of the HRQF matrix - $K \times K$, the input ordering - $input\_ordering$}
\KwOut{The best possible FSD complexity - $best\_dec\_cmplxty$ and its corresponding ordering - $best\_ordering$}

- $current\_ordering = input\_ordering$

%- $best\_ordering = input\_ordering$

- $best\_dec\_cmplxty = K$;  $i = 1$

%\If{$K = 1$}
%{
%- return
%}

%% Loop Here!!!!
\Repeat{$i = K$}
{

%\textit{Shift all the zero entries in the first row next to the first element}
- Shift all the zero entries in the first row next to the first element
%\For{$t = 2$ to $K$}
%{
%- Let the next zero entry from cursor ($t$) in the 1st row be in the $j$-th column.
%
%- Swap the $t$-th column, row with the $j$-th column, row.
%
%- Update $current\_ordering$
%}
%

- Let the number of zeros in the first row be $num\_zero\_cols$

- $grp\_size = 1$; $cur\_zero\_col = 2$. (Var under consideration)

\If{$num\_zero\_cols = 0$}
{
- $dec\_cmplxty = K$
}
\Else
{
\Repeat{$cur\_zero\_col \leq grp\_size + num\_zero\_cols$}
{
- $flag = 0$

- Let no. of consecutive zeros in $cur\_zero\_col$ be $n$.

\If{$n < grp\_size$}
{
- Move the $cur\_zero\_col$-th variable to the end (Move the corresponding row and column to the end and moving the rest of the rows and columns upwards).
%- Move the $cur\_zero\_col$-th row to the last row while moving the rest of the rows upwards.
%
%- Move the $cur\_zero\_col$-th column to the end while moving the rest of the columns to the left.

- Update $current\_ordering$.

- $num\_zero\_cols = num\_zero\_cols - 1$
}

\If{$n \geq cur\_zero\_col - 1$}
{
- $grp\_size = cur\_zero\_col - 1$

\textbf{Marker}:

- Find the next zero along the $cur\_zero\_col$ column from the $cur\_zero\_col$ row.

- Let the next zero be found in the $j$-th row.

\For{$t=1$ to $grp\_size$}
{
\If{$m_{jt} \neq 0$}
{
- $flag = 1$
}
}
\If{$flag = 1$}
{
- Move up the $j$-th variable to the $grp\_size+1$-th position.
- Update $current\_ordering$.

- $grp\_size = grp\_size + 1$

- $cur\_zero\_col = cur\_zero\_col + 1$

- jump back to \textbf{Marker}
}

\If{$flag = 0$}
{
- $cur\_zero\_col = cur\_zero\_col + 1$ (No more zeros left in this column to check)
}
}
}

- $top\_hrqf\_matrix$ = upper left $grp\_size \times grp\_size$ matrix. (With ordering $top\_hrqf\_ordering$).
%- Let the upper left $grp\_size \times grp\_size$ matrix be denoted by $top\_hrqf\_matrix$ and the corresponding ordering be $top\_hrqf\_ordering$.

- $bot\_hrqf\_matrix$ = square matrix from row, column = $grp\_size+1$ to row, column = $grp\_size+num\_zero\_cols$. (With ordering $bot\_hrqf\_ordering$).
%- Let the square matrix from row, column = $grp\_size+1$ to row, column = $grp\_size+num\_zero\_cols$ be denoted by $bot\_hrqf\_matrix$ and the corresponding ordering be $bot\_hrqf\_ordering$.

\textit{Run the current algorithm on the top and bottom matrices}

- $\left[ top\_dec\_cmplxity, best\_top\_ordering\right]  = order\_hrqf\left( top\_hrqf\_matrix, top\_hrqf\_size, top\_ordering\right)$

- $\left[ bot\_dec\_cmplxity, best\_bot\_ordering\right]  = order\_hrqf\left( bot\_hrqf\_matrix, bot\_hrqf\_size, bot\_ordering\right)$

- Update $current\_ordering$ with $best\_top\_ordering$ and $best\_bot\_ordering$.

- Number of variables conditioned: $cond\_vars = K - grp\_size - num\_zero\_cols$.

- $dec\_cmplxty = cond\_vars + max\left\lbrace top\_dec\_cmplxity, bot\_dec\_cmplxity\right\rbrace $

}
\If{$best\_dec\_cmplxty > dec\_cmplxty$}
{
- $best\_dec\_cmplxty = dec\_cmplxty$

- $best\_ordering = current\_ordering$
}

- Circularly shift the variables (rows and columns)

- Update $current\_ordering$

- $i = i + 1$

}
%% End loop here

%%%%%%%%%%%%%%%%%%%%%%%%%%%%%%%%%%%%%%%%%%%%%%%%%%%%%%%%%%%%%%%%%%%%%%%%%%%%%%%%%%%%%

\caption{The algorithm to obtain a best ordering of weight matrices - $order\_hrqf$}
\end{algorithm}
\end{table}
}

%%%%%%%%%%%%%%%%%%%
\subsubsection*{An important structural property of the HRQF matrix}
Before we delve into the proof of correctness of the algorithm, let us make a few observations regarding the structure of the HRQF matrix for various scenarios of decoding. 
Let $\mathcal{C}$ be an STBC with weight matrices $ \textbf{A}_{1}, ... , \textbf{A}_{K}$ and corresponding variables $x_{1}, ..., x_{K}$. The code $\mathcal{C}$ can be multi-group decodable, fast-decodable, fast group decodable or none of these. The structure of the upper triangular portion of the HRQF matrix for the first three of these are shown in the equations \eqref{multi_group_r_mat}, \eqref{fast_decodable_r_mat} and \eqref{fast_group_decodable_r_mat} respectively. In case the code does not allow any of these forms of decoding, the ordering of the variables is immaterial. Let $\Delta$ be a best ordering of the variables. Without loss of generality, let the first variable in this ordering be labelled $x_{1}$. Irrespective of the type of decoding provided by the code, we notice that all the variables that need to be jointly decoded with $x_{1}$ are adjacent to it in the HRQF matrix. These are followed by the set of zeros which indicate all the variables which are Hurwitz-Radon orthogonal with the previous set. These are further followed by the variables that are to be conditioned in order to obtain this group-decodable structure. In case of multi-group decodable codes, this is a null set. 

\textit{Proof of Correctness:} 
We need to show that the algorithm produces a best possible ordering for FSD complexity. 
The algorithm uses the structural property of the HRQF matrix for obtaining a best ordering. Given any ordering $\Lambda$, with the first variable as $x_{1}$, the algorithm partitions the variables into three sets - $\Lambda_{1}, \Lambda_{2}$ and $\Lambda_{3}$. In case of multi-group decoding and fast group decoding, $\Lambda_{1}$ contains all the variables that belong to the same group as $x_{1}$. In case of fast decoding, it represents the set of variables that need to be jointly decoded with $x_{1}$ after conditioning. $\Lambda_{2}$ contains all the variables HR orthogonal with $\Lambda_{1}$. $\Lambda_{3}$ contains the variables that need to be conditioned. This will be an empty set in case of multi-group decoding and hence multi-group decoding can be considered as fast decoding with no conditioned variables. The algorithm is recursively run on the sets $\Lambda_{1}$ and $\Lambda_{2}$ to further order the variables (as in case of fast group decoding).

First, we fix the first variable in the given ordering. Let it be $x_{i}$ for some $1 \leq i \leq K$. The algorithm starts with only $x_{i}$ in $\Lambda_{1}$, all variables which are HR orthogonal with $x_{i}$ in a temporary set $\Lambda_{t}$ and the rest of the variables in $\Lambda_{3}$.The variables HR orthogonal with $x_{i}$ can be easily identified as they correspond to the zero entries in the first row. For ease of manipulation, we move all the variables in the set $\Lambda_{t}$ adjacent $x_{i}$ in the ordering. This is equivalent to grouping all the zeros in the first row and placing them adjacent to the (1,1) entry, as these denote all the variables HR orthogonal with $x_{i}$. The current ordering of variables is - $\left\lbrace x_{i}, \Lambda_{t}, \Lambda_{3}\right\rbrace $. For the working of the algorithm, the following set operations are synonymous with the following matrix operations on the HRQF matrix:
\begin{itemize}
\item Moving a variable $x_{j}$ into $\Lambda_{1}$ - Suppose the variable $x_{j}$ is the $p$-th element in the current ordering, we move the $p$-th column to the $\left( \vert \Lambda_{1} \vert + 1\right)$-th column shifting the rest of the columns to the right and then we move the $p$-th row to the $\left( \vert \Lambda_{1} \vert + 1\right)$-th row shifting the rest of the rows downwards. Update the ordering accordingly. 
\item Moving a variable $x_{j}$ into $\Lambda_{3}$ - Suppose the variable $x_{j}$ is the $p$-th element in the current ordering, we move the $p$-th column to the last column shifting the rest of the columns to the left and then we move the $p$-th row to the last row shifting the rest of the rows upwards. Update the ordering accordingly.
\item Moving a variable $x_{j}$ into $\Lambda_{2}$ - No change. Only the cardinalities of $\Lambda_{2}$ and $\Lambda_{t}$ will change accordingly.   
\end{itemize}
The algorithm works in two stages.
\begin{itemize}
\item First, we find the largest $L \leq K$ such that $\vert \Lambda_{1} \vert + \vert \Lambda_{2} \vert = L.$
\item In the second stage, since the variables in $\Lambda_{1}$ and $\Lambda_{2}$ will be decoded separately, we consider the submatrices representing them as HRQF matrices of some STBC and run the first step of the algorithm on them recursively.
\end{itemize}

Let $x_{k}$ be the first variable in $\Lambda_{t}$. It is currently the second variable in the overall ordering. We now proceed to find all the variables that are HR orthogonal with $x_{k}$ and not HR orthogonal with $\Lambda_{1}$, since these will need to be jointly decoded with the variables in $\Lambda_{1}$. This can be found as follows. If any variable is HR orthogonal with $x_{k}$, it will have a zero entry in the column corresponding to $x_{k}$. Hence we traverse down the column represented by $x_{k}$ to find the next zero entry. Let it be found in the $p$-th row corresponding to the variable $x_{l}$. We also need to ensure that this variable is not HR orthogonal with $\Lambda_{1}$. So, we traverse the $p$-th row from column 1 to column $\vert \Lambda_{1} \vert$ and check for any non-zero entries. In case any of them are found, it means that this variable needs to be jointly decoded with $\Lambda_{1}$. We add $x_{l}$ to $\Lambda_{1}$. We now repeat the procedure on the variable $x_{k}$ again until all the variables have been exhausted. Since there are a finite number of variables, this process terminates. Note that it is possible that a few members from $\Lambda_{t}$ may move into $\Lambda_{1}$ during this process. These variables will be accounted for later. We now add $x_{k}$ to $\Lambda_{2}$. We have now managed to create two sets which are HR orthogonal and the current ordering of the variables is $\left\lbrace \Lambda_{1}, \Lambda_{2}, \Lambda_{t}, \Lambda_{3}\right\rbrace $. 

We now proceed to the next variable in $\Lambda_{t}$. Let this be $x_{j}$. We need to ensure that $x_{j}$ is HR orthogonal with $\Lambda_{1}$. This can be easily checked by counting the consecutive number of zeros in its column from the top of the column. If it is equal to $\vert \Lambda_{1} \vert$, then $x_{j}$ is HR orthogonal with $\Lambda_{1}$. We add $x_{j}$ to $\Lambda_{2}$ in this case. Otherwise, we move $x_{j}$ to $\Lambda_{3}$. It may so happen that $x_{j}$ may be HR orthogonal with $\Lambda_{1}$ and the already existing members of $\Lambda_{2}$. This can be checked by counting the consecutive numbers of zeros in its column from the top of the column and matching cardinalities. In such a case, we repeat the procedure of finding any variable that needs to be jointly decoded with $\Lambda_{1}$ on this variable as well. If any such variable exists, we move that variable and the pre-existing variables of $\Lambda_{2}$ into $\Lambda_{1}$. This is needed for fast group decoding scenarios. We continue this procedure till we have exhausted all the variables of $\Lambda_{t}$. We have now produced three sets $\Lambda_{1}$, $\Lambda_{2}$ and $\Lambda_{3}$. $\Lambda_{1}$ and $\Lambda_{2}$ are HR orthogonal w.r.t. one another if we condition them on the variables of $\Lambda_{3}$.

We still need to process $\Lambda_{1}$ and $\Lambda_{2}$ further. We repeat the same procedure on these variables considering them as independent HRQF matrices. If the cardinality of $\Lambda_{1}$ is the same as that of the HRQF matrix row/column size, then we stop this procedure as it cannot be ordered any further and set the FSD  complexity to $\vert \Lambda_{1} \vert$. This process terminates as the sizes of $\Lambda_{1}$ and $\Lambda_{2}$ keep reducing with each iteration and there are a finite number of variables. We calculate the FSD complexity of this ordering as - $\vert \Lambda_{3} \vert + max\left( d_{1}, d_{2} \right) $ where $d_{1}$ and $d_{2}$ are the FSD  complexities of $\Lambda_{1}$ and $\Lambda_{2}$ respectively.

Since we are dealing with fast sphere decoding and grouping the variables as sets $\Lambda_{1}, \Lambda_{2}$ and $\Lambda_{3}$ are the only possible ways to obtain a structure that will allow FS decoding, these are all the orderings we need to consider if we take $x_{i}$ as the first variable. Hence, we have exhausted all the possible orderings which can provide FS decoding with $x_{i}$ as the first variable, and a best possible FSD complexity for this case has been found by the above steps. Now, if we run through all the variables of the STBC making each one of them the first variable in turn and repeating the above procedure, we can find out a best possible FSD complexity for all possible orderings. Since all the variables have been given a chance to be the first variable, we have exhausted all possible orderings that offer FS decoding. And hence, the ordering provided by the algorithm is a best possible ordering for FSD complexity. 
\hfill $\blacksquare$

\begin{remark}
\label{hrqf_alg_remark}
Since the algorithm recursively orders each set, it is capable of even ordering variables in scenarios where any group obtained from conditioning of variables admits conditional decoding. 
\end{remark}

We now illustrate the working of the algorithm with an example. 
\begin{example}
\label{algo_fast_decode}
Consider the Silver code presented in Example \ref{silver_code}. If we order the variables as $\left[s_{1I} , s_{4I} , s_{4Q} , s_{2Q} , s_{3Q} , s_{3I} , s_{2I} , s_{1Q} \right]$, we get the following HRQF matrix and the $\textbf{R}$ matrix for this ordering:  
{\footnotesize
\begin{equation*}
%\label{algo_fast_decode_hrqf1}
\textbf{M} = \left[\begin{array}{cccccccc}
t & t & t & 0 & t & t & 0 & 0\\
t & t & 0 & t & 0 & 0 & t & t\\
t & 0 & t & t & 0 & 0 & t & t\\
0 & t & t & t & t & t & 0 & 0\\
t & 0 & 0 & t & t & 0 & t & t\\
t & 0 & 0 & t & 0 & t & t & t\\
0 & t & t & 0 & t & t & t & 0\\
0 & t & t & 0 & t & t & 0 & t\\
\end{array}\right]; 
\end{equation*}
\begin{equation*}
%\label{algo_multi_group_r1}
\textbf{R} = \left[\begin{array}{cccccccc}
t & t & t & 0 & t & t & 0 & 0\\
0 & t & t & t & t & t & t & t\\
0 & 0 & t & t & t & t & t & t\\
0 & 0 & 0 & t & t & t & t & t\\
0 & 0 & 0 & 0 & t & t & t & t\\
0 & 0 & 0 & 0 & 0 & t & t & t\\
0 & 0 & 0 & 0 & 0 & 0 & t & 0\\
0 & 0 & 0 & 0 & 0 & 0 & 0 & t\\
\end{array}\right].
\end{equation*}
}
The FSD complexity for this ordering is $M^{8}$. And this ordering does not admit fast decoding as well. 

When we run the algorithm on the given HRQF matrix, the two sets $\Lambda_{1}$ and $\Lambda_{2}$ are formed, which are HR orthogonal with each other. In this case, $\Lambda_{2} = \left\lbrace s_{2Q}, s_{2I}, s_{1Q}\right\rbrace $ and $\Lambda_{1} = \left\lbrace s_{1I}\right\rbrace $. The conditioned variables will be present in the set $\Lambda_{3} = \left\lbrace s_{4I}, s_{4Q}, s_{3Q}, s_{3I}\right\rbrace $.  The HRQF matrix at this stage is as given by \eqref{algo_fast_decode_hrqf3}. The ordering of the variables at the end of this stage is $\left[ s_{1I}, s_{2Q}, s_{2I}, s_{1Q}, s_{4I}, s_{4Q}, s_{3Q}, s_{3I}\right] $. Now, the variables from both sets $\Lambda_{1}$ and $\Lambda_{2}$ are run through the algorithm again. So, the top left $1 \times 1$ matrix and the next block diagonal $3 \times 3$ matrix are both fed to the algorithm. Since this is already the best possible ordering of these sets, the matrix $ \textbf{M}$ remains the same after this stage. And the final ordering obtained is $\left[ s_{1I}, s_{2Q}, s_{2I}, s_{1Q}, s_{4I}, s_{4Q}, s_{3Q}, s_{3I}\right] $. The $ \textbf{R}$ matrix for this ordering is given by \eqref{algo_fast_decode_r2}
\begin{equation}
\label{algo_fast_decode_hrqf3}
\textbf{M} = \left[\begin{array}{cccccccc}
t & 0 & 0 & 0 & t & t & t & t\\
0 & t & 0 & 0 & t & t & t & t\\
0 & 0 & t & 0 & t & t & t & t\\
0 & 0 & 0 & t & t & t & t & t\\
t & t & t & t & t & 0 & 0 & 0\\
t & t & t & t & 0 & t & 0 & 0\\
t & t & t & t & 0 & 0 & t & 0\\
t & t & t & t & 0 & 0 & 0 & t\\
\end{array}\right]
\end{equation}
\begin{equation}
\label{algo_fast_decode_r2}
\textbf{R} = \left[\begin{array}{cccccccc}
t & 0 & 0 & 0 & t & t & t & t\\
0 & t & 0 & 0 & t & t & t & t\\
0 & 0 & t & 0 & t & t & t & t\\
0 & 0 & 0 & t & t & t & t & t\\
0 & 0 & 0 & 0 & t & 0 & 0 & 0\\
0 & 0 & 0 & 0 & 0 & t & 0 & 0\\
0 & 0 & 0 & 0 & 0 & 0 & t & 0\\
0 & 0 & 0 & 0 & 0 & 0 & 0 & t\\
\end{array}\right]
\end{equation}
The FSD complexity for this $ \textbf{R}$ matrix is $M^{5}$ which is the best possible complexity for the Silver code. 
\end{example}

%We now consider an example with fast group decoding. 
\begin{example}
\label{algo_multi_group}
Consider the fast group decodable code presented in Example \ref{hrqf_fast_group_decode_ex}. If we order the variables as $\left[ s_{2}, s_{3}, ..., s_{10}, s_{1}, s_{11}, s_{12}, ... , s_{17}\right] $, we get the following HRQF matrix and the $\textbf{R}$ matrix: 

{\scriptsize \addtolength{\arraycolsep}{-1pt}
\begin{equation*}
%\label{algo_multi_group_hrqf1}
\textbf{M} = \left[\begin{array}{ccccccccccccccccc}
t & 0 & 0 & 0 & 0 & t & t & 0 & 0 & 0 & 0 & 0 & t & t & t & t & t\\
0 & t & 0 & 0 & 0 & t & 0 & t & t & 0 & t & 0 & 0 & t & 0 & t & t\\
0 & 0 & t & 0 & 0 & 0 & t & t & t & 0 & 0 & t & t & 0 & 0 & t & t\\
0 & 0 & 0 & t & 0 & t & t & t & 0 & 0 & t & t & 0 & 0 & t & t & 0\\
0 & 0 & 0 & 0 & t & 0 & 0 & 0 & t & 0 & t & t & t & t & t & t & 0\\
t & t & 0 & t & 0 & t & 0 & 0 & t & 0 & 0 & t & t & 0 & 0 & t & 0\\
t & 0 & t & t & 0 & 0 & t & 0 & t & 0 & t & 0 & 0 & t & 0 & t & 0\\
0 & t & t & t & 0 & 0 & 0 & t & 0 & 0 & 0 & 0 & t & t & t & t & 0\\
0 & t & t & 0 & t & t & t & 0 & t & 0 & 0 & 0 & 0 & 0 & t & t & 0\\
0 & 0 & 0 & 0 & 0 & 0 & 0 & 0 & 0 & t & 0 & 0 & 0 & 0 & 0 & 0 & 0\\
0 & t & 0 & t & t & 0 & t & 0 & 0 & 0 & t & 0 & t & 0 & 0 & t & t\\
0 & 0 & t & t & t & t & 0 & 0 & 0 & 0 & 0 & t & 0 & t & 0 & t & t\\
t & 0 & t & 0 & t & t & 0 & t & 0 & 0 & t & 0 & t & 0 & 0 & t & 0\\
t & t & 0 & 0 & t & 0 & t & t & 0 & 0 & 0 & t & 0 & t & 0 & t & 0\\
t & 0 & 0 & t & t & 0 & 0 & t & t & 0 & 0 & 0 & 0 & 0 & t & t & t\\
t & t & t & t & t & t & t & t & t & 0 & t & t & t & t & t & t & t\\
t & t & t & 0 & 0 & 0 & 0 & 0 & 0 & 0 & t & t & 0 & 0 & t & t & t\\
\end{array}\right], 
\end{equation*}
%The $\textbf{R}$ matrix for this ordering is of the form - 
\begin{equation*}
%\label{algo_multi_group_r1}
\textbf{R} = \left[\begin{array}{ccccccccccccccccc}
t & 0 & 0 & 0 & 0 & t & t & 0 & 0 & 0 & 0 & 0 & t & t & t & t & t\\
0 & t & 0 & 0 & 0 & t & 0 & t & t & 0 & t & 0 & 0 & t & 0 & t & t\\
0 & 0 & t & 0 & 0 & 0 & t & t & t & 0 & 0 & t & t & 0 & 0 & t & t\\
0 & 0 & 0 & t & 0 & t & t & t & 0 & 0 & t & t & 0 & 0 & t & t & 0\\
0 & 0 & 0 & 0 & t & 0 & 0 & 0 & t & 0 & t & t & t & t & t & t & 0\\
0 & 0 & 0 & 0 & 0 & t & t & t & t & 0 & t & t & t & t & t & t & t\\
0 & 0 & 0 & 0 & 0 & 0 & t & t & t & 0 & t & t & t & t & t & t & t\\
0 & 0 & 0 & 0 & 0 & 0 & 0 & t & t & 0 & t & t & t & t & t & t & 0\\
0 & 0 & 0 & 0 & 0 & 0 & 0 & 0 & t & 0 & t & t & t & t & t & t & t\\
0 & 0 & 0 & 0 & 0 & 0 & 0 & 0 & 0 & t & 0 & 0 & 0 & 0 & 0 & 0 & 0\\
0 & 0 & 0 & 0 & 0 & 0 & 0 & 0 & 0 & 0 & t & t & t & t & t & t & t\\
0 & 0 & 0 & 0 & 0 & 0 & 0 & 0 & 0 & 0 & 0 & t & t & t & t & t & t\\
0 & 0 & 0 & 0 & 0 & 0 & 0 & 0 & 0 & 0 & 0 & 0 & t & t & t & t & t\\
0 & 0 & 0 & 0 & 0 & 0 & 0 & 0 & 0 & 0 & 0 & 0 & 0 & t & t & t & t\\
0 & 0 & 0 & 0 & 0 & 0 & 0 & 0 & 0 & 0 & 0 & 0 & 0 & 0 & t & t & t\\
0 & 0 & 0 & 0 & 0 & 0 & 0 & 0 & 0 & 0 & 0 & 0 & 0 & 0 & 0 & t & t\\
0 & 0 & 0 & 0 & 0 & 0 & 0 & 0 & 0 & 0 & 0 & 0 & 0 & 0 & 0 & 0 & t\\
\end{array}\right].
\end{equation*}
}
The FSD complexity for this $ \textbf{R}$ matrix is $M^{13}$ but the best possible FSD complexity is $M^{12}$. 
When we run the algorithm on the given HRQF matrix, the two sets $\Lambda_{1}$ and $\Lambda_{2}$ are formed, which are HR orthogonal with each other. In this case, $\Lambda_{1} = \left\lbrace s_{2}, s_{3}, s_{4}, s_{5}, s_{6}, s_{9}, s_{10}, s_{11}, s_{12}, s_{7}, s_{8}, s_{13}, s_{14}, s_{15}, s_{16}\right\rbrace $ and $\Lambda_{2} = \left\lbrace s_{1}\right\rbrace $. The set $\Lambda_{3}$ is empty as this provides a group decoding scenario. Now, the variables from both sets are run through the algorithm again. So, the top left $16 \times 16$ matrix and the bottom right $1 \times 1$ matrix are both fed to the algorithm. The top left matrix is ordered according to the fast decoding algorithm as presented in the previous example. Since the bottom right matrix is a $1 \times 1$ matrix, it is returned without change. The final ordering of variables obtained is $\left[ s_{2}, s_{3}, s_{4}, s_{5}, s_{6}, s_{7}, s_{8}, s_{13}, s_{14}, s_{15}, s_{16}, s_{9}, s_{10}, s_{11}, s_{12}, s_{1}\right] $. The $\textbf{M}$ and the $\textbf{R}$ matrix for this ordering is as shown below.

{\scriptsize \addtolength{\arraycolsep}{-1pt}
\begin{equation*}
%\label{algo_multi_group_hrqf2}
\textbf{M} = \left[\begin{array}{ccccccccccccccccc}
t & 0 & 0 & 0 & 0 & t & t & t & t & t & t & t & 0 & 0 & 0 & 0 & 0\\
0 & t & 0 & 0 & 0 & t & 0 & 0 & t & 0 & t & t & t & t & t & 0 & 0\\
0 & 0 & t & 0 & 0 & 0 & t & t & 0 & 0 & t & t & t & t & 0 & t & 0\\
0 & 0 & 0 & t & 0 & t & t & 0 & 0 & t & t & 0 & t & 0 & t & t & 0\\
0 & 0 & 0 & 0 & t & 0 & 0 & t & t & t & t & 0 & 0 & t & t & t & 0\\
t & t & 0 & t & 0 & t & 0 & t & 0 & 0 & t & 0 & 0 & t & 0 & t & 0\\
t & 0 & t & t & 0 & 0 & t & 0 & t & 0 & t & 0 & 0 & t & t & 0 & 0\\
t & 0 & t & 0 & t & t & 0 & t & 0 & 0 & t & 0 & t & 0 & t & 0 & 0\\
t & t & 0 & 0 & t & 0 & t & 0 & t & 0 & t & 0 & t & 0 & 0 & t & 0\\
t & 0 & 0 & t & t & 0 & 0 & 0 & 0 & t & t & t & t & t & 0 & 0 & 0\\
t & t & t & t & t & t & t & t & t & t & t & t & t & t & t & t & 0\\
t & t & t & 0 & 0 & 0 & 0 & 0 & 0 & t & t & t & 0 & 0 & t & t & 0\\
0 & t & t & t & 0 & 0 & 0 & t & t & t & t & 0 & t & 0 & 0 & 0 & 0\\
0 & t & t & 0 & t & t & t & 0 & 0 & t & t & 0 & 0 & t & 0 & 0 & 0\\
0 & t & 0 & t & t & 0 & t & t & 0 & 0 & t & t & 0 & 0 & t & 0 & 0\\
0 & 0 & t & t & t & t & 0 & 0 & t & 0 & t & t & 0 & 0 & 0 & t & 0\\
0 & 0 & 0 & 0 & 0 & 0 & 0 & 0 & 0 & 0 & 0 & 0 & 0 & 0 & 0 & 0 & 0\\
\end{array}\right], 
\end{equation*}
\begin{equation*}
%\label{algo_multi_group_r1}
\textbf{R} = \left[\begin{array}{ccccccccccccccccc}
t & 0 & 0 & 0 & 0 & t & t & t & t & t & t & t & 0 & 0 & 0 & 0 & 0\\
0 & t & 0 & 0 & 0 & t & 0 & 0 & t & 0 & t & t & t & t & t & 0 & 0\\
0 & 0 & t & 0 & 0 & 0 & t & t & 0 & 0 & t & t & t & t & 0 & t & 0\\
0 & 0 & 0 & t & 0 & t & t & 0 & 0 & t & t & 0 & t & 0 & t & t & 0\\
0 & 0 & 0 & 0 & t & 0 & 0 & t & t & t & t & 0 & 0 & t & t & t & 0\\
0 & 0 & 0 & 0 & 0 & t & 0 & t & 0 & 0 & t & 0 & 0 & t & 0 & t & 0\\
0 & 0 & 0 & 0 & 0 & 0 & t & 0 & t & 0 & t & 0 & 0 & t & t & 0 & 0\\
0 & 0 & 0 & 0 & 0 & 0 & 0 & t & 0 & 0 & t & 0 & t & 0 & t & 0 & 0\\
0 & 0 & 0 & 0 & 0 & 0 & 0 & 0 & t & 0 & t & 0 & t & 0 & 0 & t & 0\\
0 & 0 & 0 & 0 & 0 & 0 & 0 & 0 & 0 & t & t & t & t & t & 0 & 0 & 0\\
0 & 0 & 0 & 0 & 0 & 0 & 0 & 0 & 0 & 0 & t & t & t & t & t & t & 0\\
0 & 0 & 0 & 0 & 0 & 0 & 0 & 0 & 0 & 0 & 0 & t & 0 & 0 & t & t & 0\\
0 & 0 & 0 & 0 & 0 & 0 & 0 & 0 & 0 & 0 & 0 & 0 & t & 0 & 0 & 0 & 0\\
0 & 0 & 0 & 0 & 0 & 0 & 0 & 0 & 0 & 0 & 0 & 0 & 0 & t & 0 & 0 & 0\\
0 & 0 & 0 & 0 & 0 & 0 & 0 & 0 & 0 & 0 & 0 & 0 & 0 & 0 & t & 0 & 0\\
0 & 0 & 0 & 0 & 0 & 0 & 0 & 0 & 0 & 0 & 0 & 0 & 0 & 0 & 0 & t & 0\\
0 & 0 & 0 & 0 & 0 & 0 & 0 & 0 & 0 & 0 & 0 & 0 & 0 & 0 & 0 & 0 & 0\\
\end{array}\right]. 
\end{equation*}
}
This ordering gives us the FSD complexity of $M^{12}$. 
\end{example}
%%%%%%%%%%%%%%%%%%%%%%%%%%%%%%%%%%%%%%%%%%%%%%%%%%%%

\section{Conclusion}
\label{sec5}

In this paper we have analysed the FSD complexity of an STBC using quadratic forms. We have shown that the HRQF completely categorizes the FSD complexity of an STBC and hence it is independent of the channel and the number of receive antennas. 
We have provided an algorithm to obtain a best ordering of weight matrices to get the best decoding performance from the code. 
%%%%%%%%%%%%%%%%%%%%%%%%%%%%%%%%%%%%%%%%%%%%%

\section*{Acknowledgements}
This work was supported partly by the DRDO-IISc program on Advanced Research in Mathematical Engineering through a research grant, and partly by the INAE Chair Professorship grant to B. S. Rajan.


\begin{thebibliography}{1}

\bibitem{HaH}
B. Hassibi and B. Hochwald, ``High-rate codes that are linear in space and time,`` IEEE Trans. Inf. Theory, vol. 48, no. 7, pp. 1804-1824, July 2002.

\bibitem{ViB}
E. Viterbo and J. Boutros, ``A Universal Lattice Code Decoder for Fading Channels,`` IEEE Trans. Inf. Theory, vol. 45, no. 5, pp. 1639-1642, July 1999.

\bibitem{BHV}
E. Biglieri, Y. Hong and E. Viterbo, ``On Fast-Decodable Space-Time Block Codes,`` IEEE Trans. Inf. Theory, vol. 55, no. 2, pp. 524-530, Feb. 2009.

\bibitem{UnM}
T. Unger and N. Markin, ``Quadratic forms and space-time block codes from generalized quaternion and biquaternion algebras,`` available online at arXiv, arXiv:0807.0199v4 [cs.IT].

\bibitem{TJC}
V. Tarokh, H. Jafarkhani and A. R. Calderbank, ``Space-Time Block Codes from Orthogonal Designs,`` IEEE Trans. Inf. Theory, vol. 45, no. 5, pp. 1456-1467, July 1999.

\bibitem{Li}
X.B. Liang, ``Orthogonal Designs with Maximal Rates,`` IEEE Trans. Inf. Theory, vol.49, no. 10, pp. 2468-2503, Oct. 2003. 

\bibitem{TiH}
O. Tirkkonen and A. Hottinen, ``Square-Matrix Embeddable Space-Time Block Codes for Complex Signal Constellations,`` IEEE Trans. Inf. Theory, vol. 48, no. 2, pp. 384-395, Feb. 2002.

\bibitem{DYT}
D. N. Dao, C. Yuen, C. Tellambura, Y. L. Guan and T. T. Tjhung, ``Four-Group Decodable Space-Time Block Codes,`` IEEE Trans. Signal Processing, vol. 56, no. 1, pp. 424-430, Jan. 2008.

\bibitem{KaR}
S. Karmakar and B. S. Rajan, ``Multigroup Decodable STBCs From Clifford Algebras,`` IEEE Trans. Inf. Theory , vol. 55, no. 1, pp. 223-231, Jan. 2009.

\bibitem{KaR1}
S. Karmakar and B. S. Rajan, ``High-rate, Multi-Symbol-Decodable STBCs from Clifford Algebras,`` IEEE Transactions on Inf. Theory, vol. 55, no. 06, pp. 2682-2695, Jun. 2009.

\bibitem{VHO}
R. Vehkalahti, C. Hollanti and F. Oggier, ``Fast-Decodable Asymmetric Space-Time Codes from Division Algebras,`` available online at arXiv, arXiv:1010.5644v1 [cs.IT].

\bibitem{PaR}
K. P. Srinath and B. S. Rajan, ``Low ML-Decoding Complexity, Large Coding Gain, Full-Rate, Full-Diversity STBCs for 2x2 and 4x2 MIMO Systems,`` IEEE Journal of Selected Topics in Signal Processing: Special issue on Managing Complexity in Multiuser MIMO Systems, vol. 3, no. 6, pp. 916-927, Dec. 2009.

\bibitem{SiB}
M. O. Sinnokrot and J. Barry, ``Fast Maximum-Likelihood Decoding of the Golden Code,`` IEEE Transactions on Wireless Commun., vol. 9, no. 1, pp. 26-31, Jan. 2010.

\bibitem{RGYS}
T. P. Ren, Y. L. Guan, C. Yuen and R. J. Shen, ``Fast-Group-Decodable Space-Time Block Code,`` Proceedings IEEE Information Theory Workshop, (ITW 2010), Cairo, Egypt, Jan. 6-8, 2010, available online at http://www1.i2r.a-star.edu.sg/cyuen/publications.html.

\bibitem{La}
T.Y. Lam, Introduction to Quadratic Forms over Fields. Providence, RI: American Mathematical Society, 2005, Graduate Studies in Mathematics no. 67.

\bibitem{KhR}
Z. Ali Khan Md., and B. S. Rajan, ``Single Symbol Maximum Likelihood Decodable Linear STBCs,`` IEEE Trans. Inf. Theory, vol. 52, no. 5, pp. 2062-2091, May 2006.

\bibitem{RaR}
G. S. Rajan and B. S. Rajan, ``Multi-group ML Decodable Collocated and Distributed Space Time Block Codes,`` IEEE Trans. Inf. Theory, vol. 56, no. 7, pp. 3221-3247, July 2010.

\bibitem{DCB}
O. Damen, A. Chkeif, and J.C. Belfiore, ``Lattice Code Decoder for Space-Time Codes,`` IEEE Communication Letters, vol. 4, no. 5, pp. 161-163, May 2000.

\bibitem{TBH}
O. Tirkkonen, A. Boariu and A. Hottinen, ``Minimal Non-Orthogonality Rate 1 Space-Time Block Code for 3+ Tx Antennas,`` Proceedings of IEEE International Symposium on Spread-Spectrum Techniques and Applications, New Jersey, pp. 429-432, Sep. 6-8, 2000.

\end{thebibliography}
\end{document}